\documentclass[11pt,a4paper]{article}
\usepackage[T1]{fontenc}
\usepackage[utf8]{inputenc}
\pdfoutput=1

\usepackage{amsmath,amssymb,amsthm}
\usepackage{booktabs}
\usepackage{authblk}
\usepackage{graphicx}
\usepackage{array}
\usepackage{tabularx}
\usepackage{longtable}
\usepackage{multirow}
\usepackage{enumitem}
\usepackage{microtype}
\usepackage{url}
\usepackage{xurl}
\usepackage{xcolor}
\usepackage{caption}
\usepackage{subcaption}
\usepackage{placeins}
\usepackage{siunitx}
\usepackage[a4paper,top=24mm,bottom=26mm,left=25mm,right=25mm]{geometry}
\usepackage{tikz}
\usetikzlibrary{arrows.meta,positioning,fit,shapes.geometric,backgrounds}
\usepackage{hyperref}

\graphicspath{{figures/}}
\hypersetup{
  colorlinks=true,
  linkcolor=blue!55!black,
  citecolor=blue!55!black,
  urlcolor=blue!55!black,
  pdfauthor={Peplluis Esteva de la Rosa; Sai Srikanth Madugula},
  pdftitle={Atomic Common-Day Invoice Clearing under Causal Daily Scheduling: Path-Enabled and Bounded-Cycle Policies}
}
\newtheorem{proposition}{Proposition}

\newcommand{\EUR}{EUR\,}
\newcommand{\PMR}{\mathrm{PMR}}
\newcommand{\IC}{\mathrm{IC}}
\newcommand{\CDG}{\mathrm{CDG}}
\newcommand{\doi}[1]{\href{https://doi.org/#1}{\nolinkurl{doi:#1}}}
\newcommand{\stableurl}[2]{\href{#1}{#2}}

\begin{document}
\emergencystretch=3em
\title{Atomic Common-Day Invoice Clearing under Causal Daily Scheduling:\\Path-Enabled and Bounded-Cycle Policies}
\author[1,2]{Peplluis Esteva de la Rosa}
\author[1]{Sai Srikanth Madugula}
\affil[1]{Woxsen University, Hyderabad, India}
\affil[2]{Universitat de Girona, Girona, Catalonia, Spain}
\date{}

\maketitle

\begin{abstract}
Late payment propagates working-capital pressure through supply networks because firms are simultaneously creditors and debtors. This paper develops an atomic-record temporal invoice-graph method for path-enabled clearing and compares it with complete-candidate bounded-cycle netting under a causal daily greedy schedule. Each invoice remains a residual record with its issue date, due date, amount, and identifier. A candidate is executable through source capacity active on every supporting edge on one common day. A non-bilateral two-edge path reduces two invoice legs, creates a direct settlement instruction between the endpoints, and preserves net positions for all participants on the combined invoice-plus-instruction state; payable-mass reduction is distinguished from invoice compression. The empirical sequence contains 749,952 invoices issued from 2012 through 2023, totalling \EUR99.705 billion, and maintains one rolling state across annual boundaries so bridge invoices are introduced once and residual balances continue. Path clearing reduces 48.202\% of issue-cohort mass, compared with 43.347\% for length-eight cycle netting, an advantage of \EUR4.841 billion and 4.855 percentage points. It leads materially in ten cohorts, is practically tied in 2013, and trails in 2012. The result survives alternative ordering, mixed path-cycle policies, cycle-length sensitivity, component resampling, acceleration constraints, and fragment replay. In reciprocal 2022, where almost all local paths lie in the cyclic core, path clearing still leads. Tractable full-information linear programs show path advantage and path-cycle complementarity. The findings establish a policy-level benefit, not global optimality or welfare dominance, and motivate future asynchronous agent-to-agent clearing built around deterministic common-day verification, private mandates, reservations, explicit consent, and atomic commit.
\end{abstract}


\vspace{1em}

\section{Introduction}\label{sec:introduction}

Late payment links financial timing to procurement, production continuity, supplier reliability, and network resilience. A firm awaiting a receivable may simultaneously owe wages, taxes, logistics providers, and suppliers. Payment delay can therefore propagate through production relations even when every originating invoice is bilateral. Financing instruments respond by introducing funds or reallocating credit risk. Multilateral clearing uses a different lever: it coordinates obligations already present in the network so that less gross payable mass must ultimately be discharged through cash transfers.

A directed invoice graph represents firms as nodes and invoices as debtor-to-creditor edges. Cycle netting subtracts a common amount from every edge of a directed circuit and preserves each participant's net position. The closure requirement makes the operation identifiable and consentable, but it excludes open chains. If $A$ owes $B$ and $B$ owes $C$, a two-edge path operation can reduce both invoice legs and create a settlement instruction from $A$ to $C$. For a non-bilateral amount $q$, residual invoice mass falls by $2q$, instruction mass rises by $q$, and post-instruction payable mass falls by $q$. Reciprocal $A\rightarrow B\rightarrow A$ operations create no instruction and reduce payable mass by $2q$.

The central empirical question is whether enlarging the admissible clearing move from closed circuits to local two-edge paths produces additional settlement relief, and whether that advantage persists when decisions must be made with only the information available at the time of execution. The study therefore compares two operational primitives on the same invoice networks. The first is bounded cycle netting, which removes a common amount around a directed circuit. The second is path-enabled clearing, which transforms a local chain $A\rightarrow B\rightarrow C$ by reducing the matched portions of the two invoice obligations and, when $A\neq C$, replacing them with a settlement instruction from $A$ to $C$. Both mechanisms preserve participant net positions, but they exploit different local structures and impose different coordination requirements.

The comparison is examined at two distinct analytical levels. The principal annual experiment uses the \emph{causal daily greedy schedule} (CDG). On calendar day $t$, both policies observe exactly the same invoices already issued and not yet expired, execute their respective deterministic candidate rules to a daily fixed point, and carry residual records forward. CDG therefore represents the operational question faced by a live clearing service: how much settlement relief can each move set realize when future arrivals are unknown and current decisions irreversibly consume invoice capacity? CDG is feasible without future-information leakage, but it is not claimed to be globally optimal.

The annual evidence is generated as one rolling causal state stream for each policy rather than as independently restarted yearly experiments. January and February records of year $y$ first become visible during the terminal bridge of cohort $y-1$. Any part consumed there is permanently removed; only the residual balance continues into the remainder of year $y$, together with records issued from March through December. At the next boundary, residual year-$y$ records meet the actual January--February arrivals of year $y+1$. After that bridge, unresolved year-$y$ records are closed for cohort accounting and only residual year-$(y+1)$ records continue. Every source identifier is therefore introduced once, cumulative consumption can never exceed its original amount, and no bridge balance is restored when the next annual row is calculated.

The denominator for cohort $y$ remains the original mass of all invoices issued in $y$, including its January--February records. PMR generated by a mixed-year operation is partitioned across issue cohorts from the consumed source fragments, so one physical operation is neither omitted nor counted twice. The rolling design therefore measures late-year clearing opportunity without restoring bridge balances or allowing the same invoice mass to support two annual results.

A second and deliberately separate level uses full-information optimization on tractable subgraphs. There the complete set of invoice arrivals and maturity intervals is available to continuous path-only, cycle-only, and mixed linear programs. These models are not alternative annual implementations and are not used to replace the causal results. They provide optimization references for a different question: how much PMR could the respective move sets support when temporal foresight and joint allocation are available? Because the realized CDG sequence is itself a feasible schedule, a genuine full-information optimum must weakly dominate it. The tractable optimization experiments therefore provide a principled reference for heuristic gaps and move-set complementarity rather than a competing operational regime.

A meaningful comparison at either level requires temporal feasibility to be defined on the underlying obligations rather than on an annual aggregate graph. Each retained source invoice is therefore represented as an atomic residual record with its own debtor, creditor, issue date, due date, amount, status, and stable identifier. For edge $(u,v)$ on day $t$, available capacity is the sum of residual amounts of records that have already been issued and have not passed their due dates. For any candidate path or circuit $F$, executable capacity on day $t$ is $\delta_F(t;s)=\min_{e\in F}c_e(t;s)$, and the maximum common-day capacity under full information is $\delta_F^*(s)=\max_{t\in\mathcal T}\delta_F(t;s)$. Thus, an operation can use only invoice fragments that coexist on a common execution day. The same atomic accounting is used in CDG with $t$ fixed to the current calendar day. This common representation makes the causal and optimization analyses comparable while preserving source-fragment provenance for independent replay.

The causal setting is also the relevant starting point for future decentralized implementation, but CDG itself should not be confused with an agent architecture. CDG is centralized: it observes the platform-wide active graph, ranks available candidates, and synchronizes execution by day. The path primitive is more local than that scheduler. A fixed $A\rightarrow B\rightarrow C$ proposal is supported by two incident obligation relationships and involves at most three firms, whereas a length-$k$ circuit requires discovery and coordination across $k$ supporting relationships. A future agent-based system could therefore use the same atomic common-day predicate, accounting rules, authorization constraints, and replay logic through asynchronous local proposal, reservation, consent, and commit processes without reproducing CDG's global ranking mechanism. Whether this locality ultimately yields lower communication cost, higher acceptance, stronger privacy, or better decentralized performance remains a research question rather than an assumption of the present experiments.

The paper addresses four questions:
\begin{enumerate}[label=RQ\arabic*.,leftmargin=2.2em]
\item Under the same rolling causal schedule, how much issue-cohort PMR does the evaluated path policy achieve relative to complete-candidate length-eight cycle netting across 2012--2023?
\item How sensitive is the result to greedy ordering, cycle-length bounds, mixed path--cycle execution, reciprocal cancellation, strongly connected structure, network concentration, payer acceleration, and bridge length?
\item What do tractable full-information optimization models reveal about heuristic gaps and the relative or complementary value of the two move sets?
\item Why is the local path primitive a promising basis for future asynchronous agent-to-agent clearing, and what must be tested before a decentralized advantage can be claimed?
\end{enumerate}

The contributions are fourfold. First, the paper gives an atomic source-record model with exact common-day capacity and an accounting theorem that distinguishes invoice compression from post-instruction PMR. Second, it provides a like-for-like causal daily comparison with PMR-aligned candidate ranking, mixed-policy and randomized-order sensitivities, cycle-length analysis, and full-information linear bounds on 256 tractable instances. Third, it evaluates a longitudinal rolling sequence in which every invoice enters once, January--February residuals carry across annual boundaries without restoration, and mixed-year PMR is attributed at fragment level; the empirical analysis also covers SCC location, reciprocal paths, acceleration, component concentration, and provenance. Fourth, it supplies independent source-fragment replay and a public reference implementation, then develops a bounded research agenda for event-driven multi-agent execution without claiming an evaluated prototype.

The principal claim is policy-specific. Across the rolling 2012--2023 sequence, path-enabled clearing reduces 48.202\% of issue-year mass and complete-candidate length-eight cycle netting reduces 43.347\%, an advantage of \EUR4.841 billion and 4.855 percentage points. The path policy leads materially in ten cohorts, the 2013 difference is a practical tie at $-0.0015$ percentage points, and cycle netting leads in the small 2012 cohort. The complete-bridge 2012--2022 subset gives the same aggregate direction. The magnitude depends on ordering, topology, source coverage, acceleration, legal discharge, and workload; it is not a theorem that paths globally dominate cycles or that PMR equals liquidity, welfare, or financing-cost savings.

\section{Related literature and positioning}\label{sec:literature}

\subsection{A connected survey of the problem}\label{sec:survey}

The literature relevant to invoice clearing is broad because the problem sits between accounting conservation, temporal scheduling, network structure, and governed execution. These traditions are often discussed separately, yet they can be read as successive answers to five questions. What quantity must be conserved? Which underlying claims are changed? When are those claims simultaneously available? What local or global transformation is permitted? Finally, who can authorize, verify, and reverse the resulting settlement? A model that answers only the first question may simplify a liability matrix but remain too abstract for invoice operations. A model that answers all five must connect source records, time, accounting, optimization, and institutional control.

The historical trend is therefore not simply from small to large networks or from heuristics to optimization. It is a movement from balance-level transformations toward claim-level and event-level execution. Financial-network models provide conservation and systemic interpretation. Production-network studies explain why interfirm obligations matter operationally. Debt-clearing algorithms turn balances into named operations. Temporal graphs prevent infeasible combinations across time. Supply-chain finance supplies the managerial and legal context. Digital and agentic systems add discovery, authorization, evidence, and distributed coordination. The contribution of this paper is to make those layers meet at one executable primitive: an atomic, common-day path or cycle operation whose source fragments, accounting effect, and replay evidence are explicit.

This synthesis also clarifies why an invoice network is not adequately characterized by a single weighted adjacency matrix. It is a directed temporal multigraph in which many records may support one ordered firm pair, active intervals differ, edge values are highly concentrated, and consuming one source record changes capacity on every day of its remaining interval. The network combines sparse topology with dense economic dependence on a small number of firms and edges. Cyclic structure matters, but so do temporal overlap, record multiplicity, maturity dispersion, validation status, and competition for bottleneck invoices. The empirical characterization in Section~\ref{sec:topology} is designed around these properties rather than around density alone.

\subsection{Debt simplification, financial clearing, and production networks}\label{sec:financial-clearing}

Verhoeff's formulation of multiple-debt settlement isolates the central conservation idea: gross obligations may be transformed while each participant's final balance remains unchanged \cite{verhoeff2004}. That insight is foundational for the present paper, but an undated balance matrix does not say which invoice is reduced, whether two claims coexist, or what legal record replaces a redirected obligation. The proposed model retains Verhoeff's conservation requirement while moving it to the combined state of residual invoices and generated settlement instructions.

Eisenberg and Noe established the canonical fixed-point model of clearing payments under limited liability and proportional repayment \cite{eisenberg2001}. Their framework asks how much is paid when firms or banks cannot meet all nominal liabilities, and it provides a rigorous system-level notion of payment feasibility. The present problem is earlier in the lifecycle: it seeks to compress solvent, due-dated trade obligations before default rather than allocate losses after insolvency. Rogers and Veraart show how bankruptcy costs and rescue possibilities alter clearing outcomes and uniqueness \cite{rogers2013}; this reinforces the point that institutional rules belong inside the model rather than being treated as implementation detail.

Elliott, Golub, and Jackson demonstrate that the pattern and concentration of financial interdependence shape the propagation of distress \cite{elliott2014}. Their analysis helps explain why gross bilateral obligations cannot be evaluated independently, even when every invoice is legally separate. Acemoglu, Ozdaglar, and Tahbaz-Salehi add an important non-monotonic perspective: greater interconnection can absorb sufficiently small shocks yet amplify larger disturbances once network losses cross critical thresholds \cite{acemoglu2015}. Glasserman and Young subsequently organize the major contagion channels and emphasize that network claims about stability depend on the mechanism being studied \cite{glasserman2016}. For invoice clearing, this literature supports network-level measurement but also cautions against interpreting PMR as a direct measure of solvency or systemic-risk reduction.

Governance enters even before default. Csoka and Herings show that decentralized clearing arrangements can produce different outcomes despite respecting common accounting constraints \cite{csoka2018}. Their contribution is especially relevant here because a mathematically admissible path is not automatically an accepted settlement. Consent, information, strategic timing, and the organization of proposal rights can change which feasible operations are realized. This paper therefore separates the deterministic economic kernel from the future decentralized mechanism that might implement it.

Production-network research supplies the operational counterpart to financial-network theory. Atalay and co-authors document a sparse, heterogeneous production structure with substantial variation in supplier and customer relationships \cite{atalay2011}. That empirical regularity anticipates the coexistence found here of a giant weakly connected trading backbone, a smaller cyclic core, and highly concentrated edge values. Carvalho shows how microeconomic disturbances can aggregate through production linkages rather than disappear through diversification \cite{carvalho2014}. Barrot and Sauvagnat identify input specificity as a reason why shocks transmit more strongly along some supplier relationships than others \cite{barrot2016}; invoice amounts alone therefore cannot capture the operational importance of every edge.

Trade-credit studies move closer to the paper's economic object. Boissay and Gropp model interfirm payment defaults and endogenous liquidity provision, explaining why a delayed receivable may constrain another firm's ability to pay \cite{boissay2013}. Their analysis treats trade credit as a propagation and insurance mechanism, whereas the present paper asks whether coordinated discharge can shorten the chain through which cash would otherwise travel. Jacobson and von Schedvin provide firm-level evidence that distress propagates through trade-credit links \cite{jacobson2015}. Ersahin, Giannetti, and Huang connect trade credit to supply-chain stability and show that payment arrangements affect resilience rather than only financing cost \cite{ersahin2024}. Acemoglu and Tahbaz-Salehi extend the production-network perspective to supply-chain disruptions and macroeconomic dynamics \cite{acemoglu2025}. Together, these works justify treating the invoice network as an industrial system with temporal consequences, while leaving open the operation-level clearing problem addressed here.

The resulting invoice graph has four distinguishing features. First, it is a multigraph: several invoices with different dates can support the same debtor--creditor edge. Second, it is interval-valued: a record is available only between issue and due dates. Third, it is stateful: every consumption changes future common-day capacity. Fourth, it is institutionally typed: validation, duplicate status, legal eligibility, and outstanding balance can determine whether a mathematically feasible fragment is executable. These features explain why the paper works from atomic records and reports topology, concentration, temporal overlap, and residual-state evolution together.

\subsection{Operational clearing, bounded circuits, and global flow}\label{sec:operational-clearing}

The closest operational comparator is the Romanian circuit-clearing architecture studied by Gavrila and Popa \cite{gavrila2021}. Their procedure constructs a company graph, identifies strongly connected components, enumerates elementary circuits, and chooses an implementation order. This is more operationally concrete than a balance-only flow because each proposed circuit has identifiable participants and a common amount. It also exposes the central limitation investigated here: closure is both the source of accounting simplicity and a restriction that excludes open chains. Increasing the length bound enlarges the cyclic search space, but it does not convert a non-cyclic intermediary relation into an eligible operation.

Kumlander's practical algorithms make ordering a first-class issue rather than a coding afterthought \cite{kumlander2010}. His later evolutionary treatment explores how alternative sequences of feasible clearings can alter terminal reduction \cite{kumlander2012}. Patcas develops graph transformations for the debt-clearing problem and emphasizes that different local rewrites may preserve balances while producing different residual structures \cite{patcas2011}. Patcas and Bartha then use evolutionary search to improve the ordering of such transformations \cite{patcasbartha2019}. These studies anticipate a central empirical fact of the present paper: the move set and the schedule interact, so a policy comparison must align economic ranking objectives and test ordering sensitivity.

Classical graph algorithms provide the cycle comparator's computational skeleton. Tarjan's linear-time decomposition isolates strongly connected regions in which directed cycles can occur \cite{tarjan1972}. Johnson's output-sensitive algorithm enumerates elementary circuits without repeatedly rediscovering the same cycle \cite{johnson1975}. These results make complete bounded-candidate enumeration possible, but they do not choose a welfare or PMR-maximizing settlement sequence. The present benchmark therefore distinguishes candidate completeness from schedule optimality.

Network-flow theory supplies a broader optimization language. Ahuja, Magnanti, and Orlin develop circulation, transshipment, and minimum-cost formulations that can route quantities globally subject to conservation and capacity constraints \cite{ahuja1993}. Such models are valuable as bounds because they reveal what simultaneous allocation could achieve. Their flexibility is also a source of distance from invoice operations: a flow solution may combine routes, counterparties, and timing decisions that do not correspond to separately named and consented transactions.

The payment-netting literature illustrates this contrast in a managerial setting. Shapiro showed early that network structure can reduce the gross payments required in international cash management \cite{shapiro1978}. Srinivasan and Kim recast the problem as network optimization and demonstrated the value of coordinated rather than pairwise payment decisions \cite{srinivasan1986}. Guntzer, Jungnickel, and Leclerc developed efficient algorithms for clearing interbank payments, where centralized control and homogeneous settlement rules make global methods more plausible \cite{guntzer1998}. Trade invoices add heterogeneous due dates, source provenance, disputes, and counterparty consent, which is why the paper retains named local operations even while using linear programs as full-information references.

Bounded coordination is not computationally innocuous. Abraham, Blum, and Sandholm show that length-constrained exchange problems become combinatorial even when every selected cycle is operationally meaningful \cite{abraham2007}. The present study accordingly separates three layers: complete enumeration of circuits through a stated bound, deterministic causal scheduling on the full annual graph, and continuous full-information optimization on tractable subgraphs. The empirical claim concerns the first two; the third diagnoses heuristic gaps and path--cycle complementarity without being extrapolated into an annual optimum.

\subsection{Temporal simultaneity, supply-chain finance, digital execution, and modern agentic AI}\label{sec:temporal-scf}

Temporal-network research explains why an annual graph can be topologically correct and operationally false. Holme and Saramaki show that aggregating time-stamped interactions can erase order, duration, and burstiness that determine what can actually propagate \cite{holme2012}. In invoice clearing, the analogous error is to combine obligations that exist in the same year but never overlap before maturity. The paper's common-day rule is therefore not cosmetic date filtering; it changes the feasible capacity of every candidate.

Casteigts and co-authors formalize time-varying graphs as graphs whose edge availability changes over time \cite{casteigts2012}. Their framework separates the footprint of all possible edges from the temporal graph that exists at a particular moment. Michail surveys temporal graph algorithms and distinguishes static reachability from journeys whose edges appear in a valid time order \cite{michail2016}. Invoice clearing imposes an even stronger synchronization requirement than a temporal journey: all supporting edges must offer the same amount on one day. The quantity $\delta_F^*(s)=\max_t\min_{e\in F}c_e(t;s)$ is thus an amount-valued certificate of simultaneous feasibility, not simply a time-respecting path test.

Supply-chain finance provides the managerial frame but pursues a different intervention. Pfohl and Gomm define supply-chain finance as coordinated planning and control of financial flows across organizational boundaries \cite{pfohl2009}. Their perspective is broad enough to include clearing, yet most subsequent instruments focus on funding, discounting, or risk transfer. Gelsomino and co-authors map the field's finance-oriented and supply-chain-oriented traditions and show that terminology often masks different objectives and actors \cite{gelsomino2016}. That distinction matters here: PMR measures a reduction in gross settlement burden, not an injection of cash.

Reverse factoring provides a useful contrast. Liebl, Hartmann, and Feisel examine its objectives, antecedents, and implementation barriers from a supply-chain perspective \cite{liebl2016}. The instrument can accelerate supplier cash by using the buyer's credit standing, but it introduces a financier and contractual infrastructure. Lekkakos and Serrano find operational benefits for small and medium-sized suppliers while also emphasizing implementation conditions \cite{lekkakos2016}. Wehinger places such instruments in the wider problem of constrained SME finance \cite{wehinger2014}. Path-enabled clearing is complementary: it can reduce pass-through payment requirements before residual obligations are financed, but it cannot replace credit where net liabilities remain.

Digital-settlement research addresses trust, provenance, and execution. Battaiola and co-authors translate invoice-factoring requirements into blockchain commitments, drawing attention to authentication and enforceable workflow states \cite{battaiola2019}. Mohammadzadeh, Dorri Nogoorani, and Munoz Tapia focus on invoice registration and the prevention of duplicate financing \cite{mohammadzadeh2021}. Ioannou and Demirel review blockchain and supply-chain finance across operations, finance, and law, showing why technical immutability does not settle questions of assignment, liability, or acceptance \cite{ioannou2022}. Babich and Hilary make the corresponding operations-management point: distributed ledgers change information and coordination possibilities, but they do not determine the economic optimization model placed on top of the ledger \cite{babich2020}.

Privacy-preserving work comes closer to network clearing. Bottazzi, Ngo, and Tsutsumi formulate multilateral trade-credit set-off using secure computation, graph anonymization, and network simplex methods \cite{bottazzi2024}. Their contribution shows that useful coordination may be possible without revealing the complete credit graph, but the model is oriented toward secure global computation rather than atomic common-day local moves. Buchman and co-authors describe a peer-to-peer multilateral clearing architecture in which participants coordinate without a conventional central operator \cite{buchman2024}. The architectural ambition is relevant, while source-record timing, path redirection, and fragment-level replay remain distinct contributions of the present paper.

Only a selected part of the modern agentic-AI literature is directly relevant. Guo and co-authors survey LLM-based multi-agent systems through agent roles, communication, coordination, and evaluation \cite{guo2024multiagents}. Their taxonomy is useful because it makes clear that fluent interaction is not the same as reliable economic execution: an agent system still needs a state model, verifiable actions, and domain-specific outcome metrics. Xu, Mak, Minaricova, and Brintrup move from general orchestration to the supply-chain domain by proposing a design methodology, system architecture, toolkit, and case study for autonomous supply chains \cite{xu2024implementation}. Their work establishes that firm-level agents can be embedded in industrial decision processes, but it does not define how due-dated financial claims should be transformed or discharged.

Recent interoperability standards provide components for such an implementation. The Agent2Agent protocol defines capability discovery, task lifecycles, asynchronous updates, and collaboration between independently implemented and internally opaque agents \cite{a2a2025}. This is well matched to discovering and tracking a clearing proposal, but A2A does not supply invoice capacity semantics, reservations over monetary records, or legal discharge. The Model Context Protocol standardizes how an agent host obtains controlled access to resources and tools, with explicit attention to authorization, consent, and tool safety \cite{mcp2025}. It could expose a firm's invoice vault, policy engine, or deterministic common-day verifier without turning those controls over to the language model; it is not itself a protocol for interfirm commitment.

The Agent Payments Protocol provides the closest modern analogue to the authorization layer required here. AP2 uses cryptographic mandates, receipts, and verifiable evidence to bind delegated intent to agent-initiated payment actions \cite{ap22026}. That architecture supports a crucial separation also adopted in this paper's research agenda: agents may discover and negotiate, while deterministic credentials and verifiers control commitment. AP2 is built around commerce and payment authorization, however, not the multi-claim discharge and issue-cohort accounting of invoice clearing. AgenticPay complements the protocol view with an evaluation framework for language-mediated buyer--seller negotiation under private constraints \cite{agenticpay2026}. Its results expose persistent failures in long-horizon strategy and constraint satisfaction, suggesting that future clearing agents should be assessed through feasibility, efficiency, fairness, timeout, and violation metrics rather than conversational quality alone.

These modern works define an emerging stack: domain agents, tool access, inter-agent messaging, mandates, and negotiation benchmarks. The paper contributes a missing economic kernel for that stack. It specifies the exact state to be committed, the common-day capacity that agents may propose, the accounting transformation created by acceptance, the source-fragment evidence needed for replay, and a causal centralized benchmark against which an asynchronous implementation can be measured. Section~\ref{sec:decentralized-agenda} develops that connection as a research agenda rather than claiming a prototype.

\subsection{Differentiation and novelty}\label{sec:positioning}

Table~\ref{tab:positioning} clarifies the distinction among the principal model classes. The strongest novelty is the combination of atomic dated claims, exact common-day amount selection, an explicit instruction layer, source-level replay, and a longitudinal like-for-like comparison under causal daily scheduling.

\begin{table}[htbp]
\centering
\caption{Positioning of the proposed method relative to major model classes. A check denotes a defining feature rather than a claim that every work in the class implements it identically.}
\label{tab:positioning}
\footnotesize
\renewcommand{\arraystretch}{1.25}
\setlength{\tabcolsep}{2.4pt}
\begin{tabularx}{\textwidth}{>{\raggedright\arraybackslash}p{0.20\textwidth}*{7}{>{\centering\arraybackslash}X}}
\toprule
Model class & Atomic dated claims & Common-day amount & Open two-edge path & Explicit instruction layer & Named local operation & Global optimization & Invoice-plus-instruction conservation \\
\midrule
Financial clearing/default models & -- & -- & sometimes & -- & -- & often & balance-level \\
Debt simplification/network flow & rarely & -- & often & implicit & usually -- & often & balance-level \\
Operational bounded-cycle netting & sometimes & rarely & -- & -- & yes & heuristic or combinatorial & invoice layer \\
Temporal-network reachability & \shortstack{time-\\stamped} & temporal order & \shortstack{not an\\accounting\\move} & -- & \shortstack{path/\\journey} & varies & -- \\
Proposed method & yes & yes & yes & yes & yes & heuristic plus LP bounds & yes \\
\bottomrule
\end{tabularx}
\end{table}

The paper therefore does not claim to invent debt simplification, temporal graphs, or distributed settlement. Its contribution is a specific operational integration: source-record-exact common-day capacity; a triadic redirection primitive with explicit accounting; comparable PMR objectives for paths and cycles; and evidence on how those policies behave across heterogeneous annual invoice networks.

\section{Atomic temporal model and accounting}\label{sec:model}

\subsection{Atomic records and edge-day capacity}\label{sec:atomic-records}

Let $V$ be the set of firms. Each retained source row is an atomic invoice record
\begin{equation}
 i=(u_i,v_i,x_i,\alpha_i,\tau_i,\sigma_i,\eta_i),
\label{eq:record}
\end{equation}
where $u_i\in V$ is the debtor, $v_i\in V$ the creditor, $x_i>0$ the original amount, $\alpha_i$ the issue date, $\tau_i\geq\alpha_i$ the due date, $\sigma_i$ a validation indicator, and $\eta_i$ a stable source-row identifier. The residual amount at state $s$ is $r_i(s)\in[0,x_i]$.

The aggregate topology $G(s)=(V,E(s))$ contains edge $(u,v)$ when at least one record from $u$ to $v$ has positive residual. Aggregation is used for candidate discovery only. For calendar day $t$, active residual edge capacity is
\begin{equation}
 c_{uv}(t;s)=\sum_{i:u_i=u,\,v_i=v} r_i(s)\mathbf{1}\{\alpha_i\leq t\leq\tau_i\}.
\label{eq:edge-day-capacity}
\end{equation}
The annual implementation evaluates every integer calendar day in the observation horizon. Capacity changes only at issue dates and the day after due dates; the tractable linear programs therefore use those event days as an exact compressed representation.

Residual invoice mass is $D(s)=\sum_i r_i(s)$. For firm $v$, let $I_s(v)$ and $O_s(v)$ be incoming and outgoing residual mass and $b_s(v)=I_s(v)-O_s(v)$ its invoice-layer net position.

\subsection{Exact common-day capacity}\label{sec:common-day}

For a fixed path or circuit with edge set $F$, matched capacity on day $t$ is
\begin{equation}
 \delta_F(t;s)=\min_{e\in F}c_e(t;s).
\label{eq:day-capacity}
\end{equation}
The maximum amount executable on one common day is
\begin{equation}
 \delta_F^*(s)=\max_{t\in\mathcal T}\delta_F(t;s),
\label{eq:common-day-capacity}
\end{equation}
with the earliest maximizing day used as a deterministic tie-breaker.

\begin{proposition}[Exact fixed-operation capacity]\label{prop:exact-capacity}
For fixed $F$ and state $s$, $\delta_F^*(s)$ is the largest amount assignable to one operation day using only residual source records active on every edge in $F$.
\end{proposition}

\noindent\textit{Proof.} On day $t$, no edge can supply more than $c_e(t;s)$, so every common amount is bounded by $\min_{e\in F}c_e(t;s)$. Maximizing over days gives an upper bound. At a maximizing day, each edge has at least $\delta_F^*(s)$ active residual capacity, so consuming that amount from active records constructs a feasible operation. \hfill$\square$

The interval condition
\begin{equation}
 \max_i\alpha_i\leq\min_i\tau_i
\label{eq:interval-characterization}
\end{equation}
is retained only as the equivalent non-empty-intersection characterization of a particular selected collection of fragments. Equation~\eqref{eq:common-day-capacity} is the operative amount-and-day rule.

When an operation is executed, records on each edge are consumed by earliest due date, earliest issue date, and source-row identifier. Splitting is allowed, but every consumed fragment must be active on the logged day.

\subsection{Settlement state, PMR, and legal interpretation}\label{sec:accounting}

A run maps initial state $s_0$ to terminal residual state $s_T$. Non-bilateral path operations also generate a settlement-instruction multiset $P$. If $P_{in}(v)$ and $P_{out}(v)$ are instruction mass entering and leaving $v$, the combined terminal net position is
\begin{equation}
 b_T^*(v)=I_{s_T}(v)+P_{in}(v)-O_{s_T}(v)-P_{out}(v).
\label{eq:combined-net}
\end{equation}
A valid run preserves $b_T^*(v)=b_{s_0}(v)$ for all firms.

Post-instruction payable mass is
\begin{equation}
 \Omega_T=D(s_T)+P_{tot},\qquad P_{tot}=\sum_{p\in P}\mathrm{amount}(p),
\label{eq:post-settlement}
\end{equation}
and the primary outcome is
\begin{equation}
 \PMR=D(s_0)-\Omega_T.
\label{eq:pmr}
\end{equation}
Invoice-layer compression is $\IC=D(s_0)-D(s_T)$. Cycle netting creates no instruction, so $\PMR=\IC$. A non-bilateral path of amount $q$ has $\IC=2q$ and $\PMR=q$; a reciprocal path has $\PMR=\IC=2q$.

\begin{proposition}[Path accounting]\label{prop:path-accounting}
Every path operation preserves combined net positions. A non-bilateral operation of amount $q$ reduces post-instruction payable mass by $q$; a reciprocal operation reduces it by $2q$.
\end{proposition}

\noindent\textit{Proof.} For $A\neq C$, invoice outflow of $A$ and invoice inflow of $C$ each fall by $q$, while an $A\rightarrow C$ instruction of $q$ replaces them. At intermediary $B$, one inflow and one outflow both fall by $q$. Residual mass falls by $2q$ and instruction mass rises by $q$. When $A=C$, opposite invoice legs fall without an instruction. \hfill$\square$

The empirical PMR interpretation assumes \emph{consented discharge upon performance}: participants authorize the redirected payment, and performance discharges the matched invoice portions. Under instruction-only mode, the algorithm identifies proposed compression but legal discharge is realized only after payment. Under novation, the same accounting identity holds but counterparty exposure changes and requires separate legal and credit analysis. Net-position preservation alone does not establish tax, insolvency, sanctions, KYC, or finality equivalence.

\subsection{Causal feasibility and the full-information optimization reference}\label{sec:dominance}

CDG is the sole annual execution policy studied in the main experiment. Its chronological operation sequence is source-record feasible: every consumed fragment is active on its execution day, no record is overconsumed, and all accounting identities hold. Let $\mathcal S_F$ denote all feasible schedules for a fixed move set over the complete horizon.

\begin{proposition}[Full-information optimum weakly dominates CDG]\label{prop:dominance}
For fixed source records, horizon, move set, and accounting objective,
\begin{equation}
 \max_{S\in\mathcal S_F}\PMR(S)\geq \PMR(\CDG).
\label{eq:dominance}
\end{equation}
\end{proposition}

\noindent\textit{Proof.} The complete CDG operation sequence is one member of $\mathcal S_F$. The maximum over $\mathcal S_F$ cannot be smaller than the value of that feasible member. \hfill$\square$

The annual-scale optimum is not computed. Continuous full-information linear programs are instead solved on tractable weak components and induced strongly connected samples. They provide upper references for divisible source amounts and reveal heuristic gaps and path--cycle complementarity without being confused with the operational CDG implementation.

\section{Policies, comparators, and optimization benchmarks}\label{sec:algorithms}

\subsection{PMR-aligned path policy}\label{sec:path-policy}

For intermediary $B$, an incoming edge $e^-=(A,B)$ and outgoing edge $e^+=(B,C)$ define a candidate. In CDG on day $t$, its matched amount is $q=\min\{c_{e^-}(t;s),c_{e^+}(t;s)\}$. The candidate score is actual PMR:
\begin{equation}
 \mathrm{score}_P(A,B,C,q)=
 \begin{cases}
 2q,&A=C,\\
 q,&A\neq C.
 \end{cases}
\label{eq:path-score}
\end{equation}
On each day, intermediaries are visited in ascending stable identifier order. At one intermediary, the policy compares the best currently active reciprocal candidate with the best non-bilateral incoming--outgoing pair, selects the larger PMR score, consumes its source fragments, and repeats to a local fixed point before advancing to the next intermediary. Ties use matched amount, bilateral indicator, endpoint identifiers, and edge identifiers. This ordering is deterministic and PMR-aligned but not globally optimal across intermediaries. Comparator fairness requires the economic score above: a reciprocal path produces PMR $2q$, whereas a non-bilateral path produces PMR $q$; the cycle score $kq$ is already PMR-aligned. This PMR-aligned score is used for every annual path result, including the accepted 2021 workbook corpus.

After selection, both edges are consumed from active atomic records using the deterministic rule in Section~\ref{sec:common-day}. Non-bilateral fragments are paired in consumption order and create instruction subrecords with inherited due date $\min\{\tau_{in},\tau_{out}\}$. Instructions are not recycled into candidate discovery; the study therefore evaluates one-stage invoice compression plus an auditable settlement layer. Recycling is a separate recursive-clearing design with different exposure and termination properties.

\subsection{Complete-candidate bounded-cycle policy}\label{sec:cycle-policy}

The comparator constructs the aggregate topology, decomposes it into strongly connected components, enumerates every elementary directed circuit through length $L$, and executes explicit circuit settlements. For circuit $\gamma$ of length $k$ and common-day amount $q$, the score is
\begin{equation}
 \mathrm{score}_C(\gamma,q)=kq,
\label{eq:cycle-score}
\end{equation}
which equals PMR because no instruction is created. The primary bound is $L=8$, following the operational emphasis on short consent coalitions and because marginal CDG gains from $L=8$ to $L=10$ are empirically negligible in the two large sensitivity years. No per-edge record cap is used.

\subsection{Causal daily, mixed, randomized, and LP benchmarks}\label{sec:regimes}

CDG advances one calendar day at a time. On day $t$, only issued records with unexpired due dates are eligible. The cycle policy ranks all currently feasible bounded circuits by PMR. The path policy applies the intermediary-local rule in Section~\ref{sec:path-policy}. Each updates the supporting atomic records and repeats to its defined daily fixed point; residual records carry forward. The policies receive exactly the same causal information, although their candidate structures and deterministic sequencing differ.

Annual execution is a single chronological stream for each method. For the first cohort, records enter on their 2012 issue dates. For every later year $y$, its January--February records have already entered during the bridge of cohort $y-1$. The residuals of those records, and only those residuals, continue from 1 March alongside newly arriving March--December invoices. Residual cohort-$y$ records then interact with January--February arrivals of $y+1$ during the terminal bridge. After the bridge, unresolved cohort-$y$ records are closed and only residual cohort-$(y+1)$ records continue. Consequently, each UID is introduced exactly once and the state transition at every boundary satisfies
\begin{equation}
 x_i=c^{m,\mathrm{bridge}}_i+r^{m,\mathrm{carry}}_i
 \qquad\text{for every bridge record }i\text{ and method }m,
\label{eq:nonreuse}
\end{equation}
where $c^{m,\mathrm{bridge}}_i$ is its physical consumption during the bridge and $r^{m,\mathrm{carry}}_i$ is the residual entering the rest of its own issue year. No consumed amount is reintroduced at face value.

Cohort accounting is distinct from physical state continuity. The denominator for cohort $y$ is all original invoice mass issued in $y$. Cycle PMR is attributed by the issue year of each consumed leg; reciprocal path PMR is attributed by leg; and the one unit of PMR from a non-bilateral path is split symmetrically between its incoming and outgoing consumed fragments. These shares sum exactly to the physical operation PMR. Conservative and liberal mixed-year bounds assign the unit wholly to neither or either cohort leg, respectively.

Three supplementary benchmarks address move-set and ordering concerns. First, a sequential mixed daily policy executes the cycle policy to a daily fixed point and then the path policy, or reverses that order. Second, randomized near-greedy runs choose uniformly from candidates scoring at least 95\% of the current best. Third, on tractable subgraphs a continuous full-information linear program allocates source-record capacity across candidate-day variables for path-only, cycle-only, or mixed move sets. Continuous divisibility makes its objective an upper reference for integer-cent schedules. The LP is not a legal implementation model; it is a diagnostic for heuristic gaps and move-set complementarity.

Randomized stability tests use 40 induced strongly connected samples from 2019, 2020, 2022, and 2023. Fifty seeds are used for each sample and policy. These repetitions test near-greedy ordering sensitivity without treating invoices as independent observations.

\subsection{Computational representation}\label{sec:complexity}

For an annual horizon of at most 404 days, edge-day capacity is stored as an integer array. Exact pair evaluation is $O(|\mathcal T|)$; circuit evaluation is $O(k|\mathcal T|)$. Lazy versioning reevaluates a candidate only when a supporting edge changes. Circuit enumeration is output-sensitive in the number of circuits. The 2020 topology contains 899,363 circuits through length eight; complete enumeration required 15.63 seconds and exact initial common-day screening 0.77 seconds in the research environment. The extension used Python 3.13.5, NetworkX 3.6.1, NumPy 2.3.5, SciPy/HiGHS for tractable LPs, and single-process C++17 for selected large preprocessing tasks on Debian 13. These timings exclude identification, consent, legal checks, payment-rail communication, and failure recovery.

\section{Data, annual cohorts, and validation}\label{sec:data}

\subsection{Source construction and analytical population}\label{sec:data-source}

The analytical sequence covers invoice issue years 2012--2023. The harmonized company-code source contains 792,689 submitted rows for 1995--2023. Records from 1995--2011 are excluded because annual coverage is sparse and discontinuous; the longitudinal 2021 slice is also excluded because the accepted 2021 evidence comes from twelve separately curated monthly workbooks. Selecting 2012--2020 and 2022--2023 from the harmonized source gives 658,405 input rows. Adding 129,158 submitted 2021 workbook rows gives 787,563 analytical input rows. Cleaning retains 749,952 atomic invoices totalling \EUR99.705016 billion.

The institutional source is a commercially sensitive Romanian invoice-reporting and clearing environment. Participation and reporting coverage are not demonstrated to be nationally representative. The paper therefore avoids ``national-scale'' effect claims. The annual cohorts differ sharply in size and source coverage; they are repeated empirical environments, not an IID panel or causal time series.

The rolling sequence uses the actual January and February arrivals following each issue year. Complete two-month follow-up is available through the 2022 cohort. The source supplied for the final 2023 cohort ends on 8 February 2024, yielding a 39-day terminal bridge; the cohort is retained because its full issue year is observed, and the equal-horizon 2012--2022 subset is reported as a coverage sensitivity. Across the 12 boundaries, 116,570 January--February records totalling \EUR16.584 billion are introduced exactly once. Of these, 115,330 records issued in 2013--2023 also belong to their own issue-year denominators; 1,240 records issued in 2024 support the final bridge but lie outside the analytical issue-year range. None is added to the preceding cohort denominator.

Identifier continuity is source-specific but explicit. Harmonized adjacent years share one anonymized company-code namespace. The 2020--2021 boundary uses the accepted 2021 workbooks linked through a one-to-one company crosswalk; the 2021--2022 boundary links 2021 CIFs into the harmonized anonymous-code namespace. Unresolved identifiers remain distinct rather than being guessed. Physical non-reuse and cohort attribution follow Section~\ref{sec:regimes}.

\subsection{Cleaning and provenance}\label{sec:cleaning}

Records are removed for nonnumeric or nonpositive amount, missing or invalid due date, negative maturity, cancellation, missing/self counterparty, or known source overlap. Long contractual maturity is not a deletion criterion. Exact fingerprints are audited but retained in the main corpus where immutable transaction identifiers are unavailable; duplicate-filtered sensitivities remove excess copies by stable first occurrence.

\begin{figure}[htbp]
\centering
\resizebox{\textwidth}{!}{%
\begin{tikzpicture}[node distance=8mm and 9mm, every node/.style={font=\small}, box/.style={draw,rounded corners,align=center,minimum height=9mm,text width=3.4cm}, arrow/.style={-{Latex[length=2mm]},thick}]
\node[box] (source) {Harmonized source\\792,689 rows\\1995--2023};
\node[box,right=of source] (select) {Select 2012--2020, 2022--2023\\exclude sparse early years and longitudinal 2021\\658,405 rows};
\node[box,below=of source] (workbooks) {Accepted 2021 monthly workbooks\\129,158 rows};
\node[box,right=of workbooks] (input) {Analytical input\\787,563 rows};
\node[box,right=of input] (clean) {Cleaning exclusions\\37,611 rows};
\node[box,right=of clean] (final) {Retained annual sequence\\749,952 atomic invoices\\\EUR99.705 bn};
\draw[arrow] (source)--(select);
\draw[arrow] (select)--(input);
\draw[arrow] (workbooks)--(input);
\draw[arrow] (input)--(clean);
\draw[arrow] (clean)--(final);
\end{tikzpicture}}
\caption{Source selection and cleaning flow. The 792,689-row harmonized source is not itself transformed directly into the final corpus because its longitudinal 2021 slice is replaced by the accepted monthly-workbook pipeline.}
\label{fig:cleaning-flow}
\end{figure}
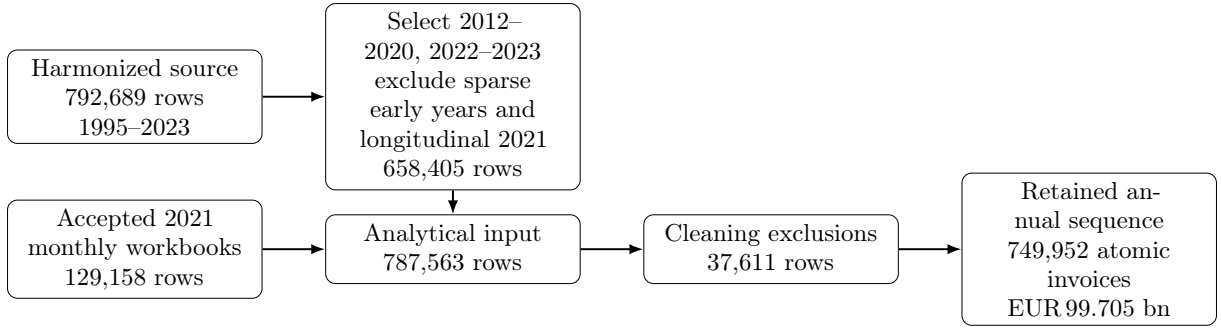

The 37,611 cleaning exclusions comprise 12,290 bad amounts, 19,146 bad due dates, 4,770 negative maturities, 1,270 cancellations, eight missing or self-party records, and 127 known overlaps. Table~\ref{tab:provenance-condensed} summarizes annual provenance; the complete audit appears in Appendix~\ref{app:provenance}.

\begin{table}[htbp]
\centering
\caption{Condensed annual provenance, scale, and bridge coverage. Retained issue-year mass is the PMR denominator. January--February records may first enter during the preceding bridge, but only their residual balances continue into their own year.}
\label{tab:provenance-condensed}
\footnotesize
\setlength{\tabcolsep}{3.0pt}
\begin{tabularx}{\textwidth}{rrrrr>{\raggedright\arraybackslash}p{0.16\textwidth}>{\raggedright\arraybackslash}X}
\toprule
Years & Input & Retained & Mass (bn) & Firms & Source & Bridge and principal comparability limit \\
\midrule
2012--2018 & 18,120 & 16,494 & 2.319 & 138--1,160 & Harmonized & Complete next-year Jan--Feb; sparse legacy coverage \\
2019--2020 & 460,782 & 445,541 & 18.444 & 25,383--63,371 & Harmonized & Complete next-year Jan--Feb; sharp coverage expansion \\
2021 & 129,158 & 121,844 & 15.153 & 14,250 & Monthly workbooks & Complete Jan--Feb 2022 via company crosswalk; separate accepted pipeline \\
2022 & 85,148 & 79,471 & 41.576 & 3,743 & Harmonized & Complete Jan--Feb 2023; highly concentrated and reciprocal \\
2023 & 94,355 & 86,602 & 22.212 & 4,460 & Harmonized & Observed follow-up through 8 Feb 2024 \\
\bottomrule
\end{tabularx}
\end{table}

\subsection{Annual network characterization}\label{sec:topology}

The annual graphs are sparse directed multigraphs with highly concentrated edge weights and many source records per economically important relationship. Static topology varies substantially across years. Reciprocity is low in 2019--2020 but rises to 68.88\% in 2022. The largest strongly connected component contains 35.11\% of 2022 firms. The ratio of all local two-edge paths to those in the cyclic core, $R_{top}$, is 1.000157 in 2022: 99.984\% of local paths lie in cyclic structure. The graph contains 48,109 circuits through length eight, of which 4,325 have positive initial common-day capacity.

\begin{figure}[htbp]
\centering
\includegraphics[width=0.93\textwidth]{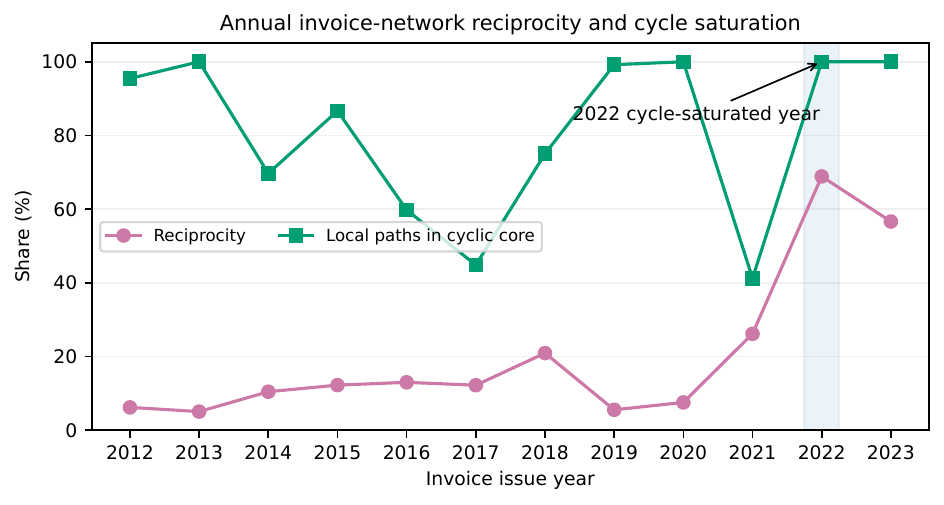}
\caption{Annual reciprocity and count-weighted local-path cycle saturation. The 2022 graph is the most cycle-saturated large cohort.}
\label{fig:annual-topology}
\end{figure}

This characterization is important because it prevents a simplistic interpretation. In 2022, path advantage cannot primarily arise from reaching an acyclic periphery: almost every local path is already in the cyclic core. The two move sets transform shared records differently, and their temporal schedules compete for bottleneck capacity.

\subsection{Validation and unit of inference}\label{sec:validation}

An independent replay program reloads every phase input and reconstructs each logged source fragment. It verifies source-record existence, operation-day activity, edge consistency, fragment sums, nonnegative residuals, instruction mass, PMR identities, issue-year attribution, and firm-level combined net positions. The 24 chronological phases for each policy---12 within-year phases and 12 terminal bridges---contain 320,579 cycle and 531,808 path fragment allocations. Every phase passes replay.

A second global audit spans phase boundaries. It records the original amount of all 751,192 introduced UIDs, accumulates physical consumption across the complete stream, and rejects either duplicate introduction or cumulative consumption above face value. All UIDs are introduced once; cycle and path consume from 268,748 and 397,470 distinct records, respectively; maximum overconsumption is zero. At each of the 12 boundaries, the aggregate identity corresponding to Eq.~\eqref{eq:nonreuse} holds exactly. The January--February bridge pool contains 116,570 records and \EUR16.584 billion: cycle clearing consumes \EUR6.628 billion and carries \EUR9.956 billion of residual mass, while path clearing consumes \EUR9.885 billion and carries \EUR6.699 billion. No bridge record is restored for the following annual phase.

Invoices are not treated as independent statistical observations. The primary repeated unit is the annual cohort, supplemented by connected-component and induced-subgraph analyses. Exact annual signs give ten path wins and two cycle wins; a one-sided sign test gives $p=0.0193$ under an exchangeable 50--50 null. This is descriptive because annual coverage and scale are heterogeneous. The 2013 difference is only $-0.0015$ percentage points and is treated as a practical tie under a 0.01-point tolerance. Component resampling is interpreted as stability and concentration analysis, not as an IID confidence interval.

\section{Results}\label{sec:results}

\subsection{Rolling causal annual comparison}\label{sec:main-results}

Table~\ref{tab:annual-main} reports issue-cohort PMR after each observed terminal bridge. Across all 12 cohorts, cycle netting reduces \EUR43.219 billion, or 43.347\% of issue-year mass. The path policy reduces \EUR48.060 billion, or 48.202\%. The aggregate difference is \EUR4.841 billion and 4.855 percentage points. Path clearing leads materially in ten cohorts, is practically tied in 2013, and trails in the small 2012 graph, which contains 339 records and seven bounded circuits.

\begin{table}[htbp]
\centering
\caption{Rolling causal daily PMR by issue cohort. January--February invoices enter once: any bridge consumption is permanent and only residual balance continues into the invoice's own issue year.}
\label{tab:annual-main}
\footnotesize
\renewcommand{\arraystretch}{1.14}
\setlength{\tabcolsep}{4.6pt}
\begin{tabular}{rrrrrr}
\toprule
Year & Mass (bn) & Bridge days & Cycle $L=8$ (\%) & Path-enabled (\%) & Difference (pp) \\
\midrule
2012 & 0.069 & 59 & 1.3407 & 0.9108 & $-0.4299$ \\
2013 & 0.119 & 59 & 0.0379 & 0.0364 & $-0.0015$ \\
2014 & 0.197 & 59 & 0.2255 & 0.7925 & $+0.5670$ \\
2015 & 0.538 & 60 & 0.2537 & 0.3549 & $+0.1012$ \\
2016 & 0.679 & 59 & 1.2021 & 1.3075 & $+0.1054$ \\
2017 & 0.257 & 59 & 3.8130 & 4.3427 & $+0.5297$ \\
2018 & 0.460 & 59 & 3.4529 & 5.0824 & $+1.6295$ \\
2019 & 4.424 & 60 & 14.1400 & 15.8582 & $+1.7182$ \\
2020 & 14.020 & 59 & 42.7905 & 47.7576 & $+4.9671$ \\
2021 & 15.153 & 59 & 36.7300 & 43.4099 & $+6.6800$ \\
2022 & 41.576 & 59 & 52.7948 & 56.4554 & $+3.6607$ \\
2023 & 22.212 & 39$^{a}$ & 40.7050 & 47.5624 & $+6.8575$ \\
\midrule
Mass-weighted & 99.705 & observed & 43.3468 & 48.2018 & $+4.8551$ \\
\bottomrule
\end{tabular}
\vspace{1mm}
\begin{minipage}{0.94\textwidth}\footnotesize
$^{a}$Available 1 January--8 February 2024. Cohorts 2012--2022 use complete January and February bridges. A difference with absolute value below 0.01 percentage points is treated as a practical tie.
\end{minipage}
\end{table}

At 31 December, before the terminal bridges, aggregate cycle and path PMR are 42.394\% and 47.331\%. The bridges add \EUR0.950 billion of cohort-attributed cycle PMR and \EUR0.869 billion of path PMR. The path lead consequently narrows from \EUR4.922 billion to \EUR4.841 billion, or from 4.936 to 4.855 percentage points. For the 2012--2022 subset, where every terminal bridge covers the complete first two months of the following year, cycle and path PMR are 44.104\% and 48.385\%, respectively, a 4.281-point path advantage. This coverage-matched subset is reported as a follow-up-horizon sensitivity.

\begin{figure}[htbp]
\centering
\includegraphics[width=0.93\textwidth]{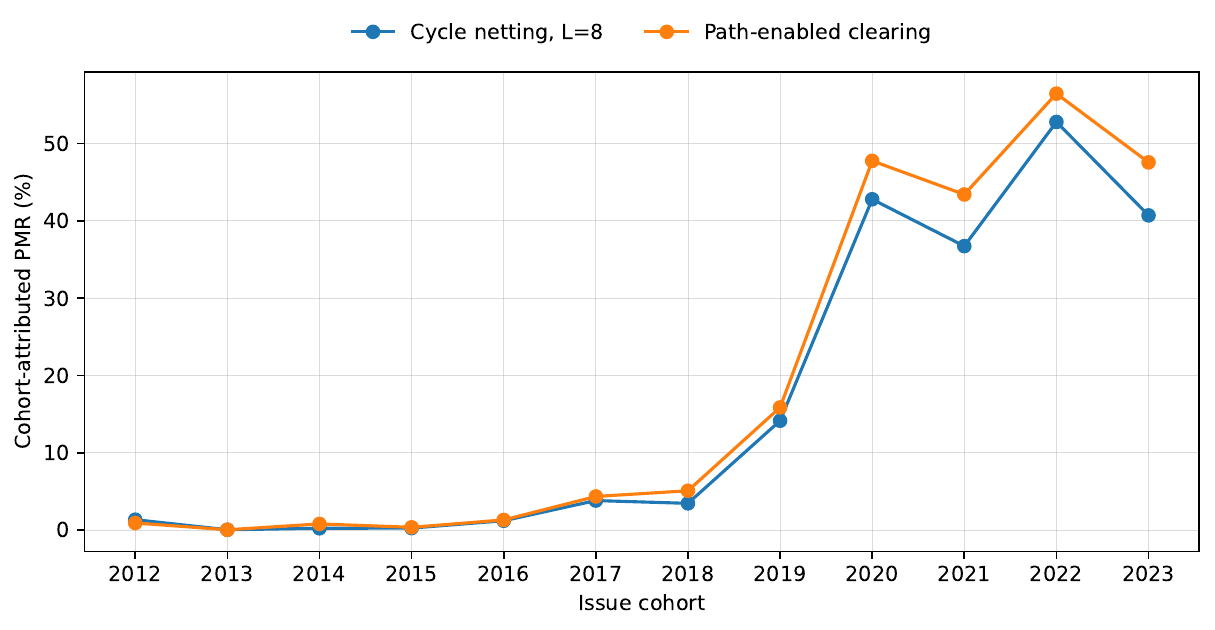}
\caption{Terminal issue-cohort PMR after each observed bridge.}
\label{fig:annual-pmr}
\end{figure}

\begin{figure}[htbp]
\centering
\includegraphics[width=0.93\textwidth]{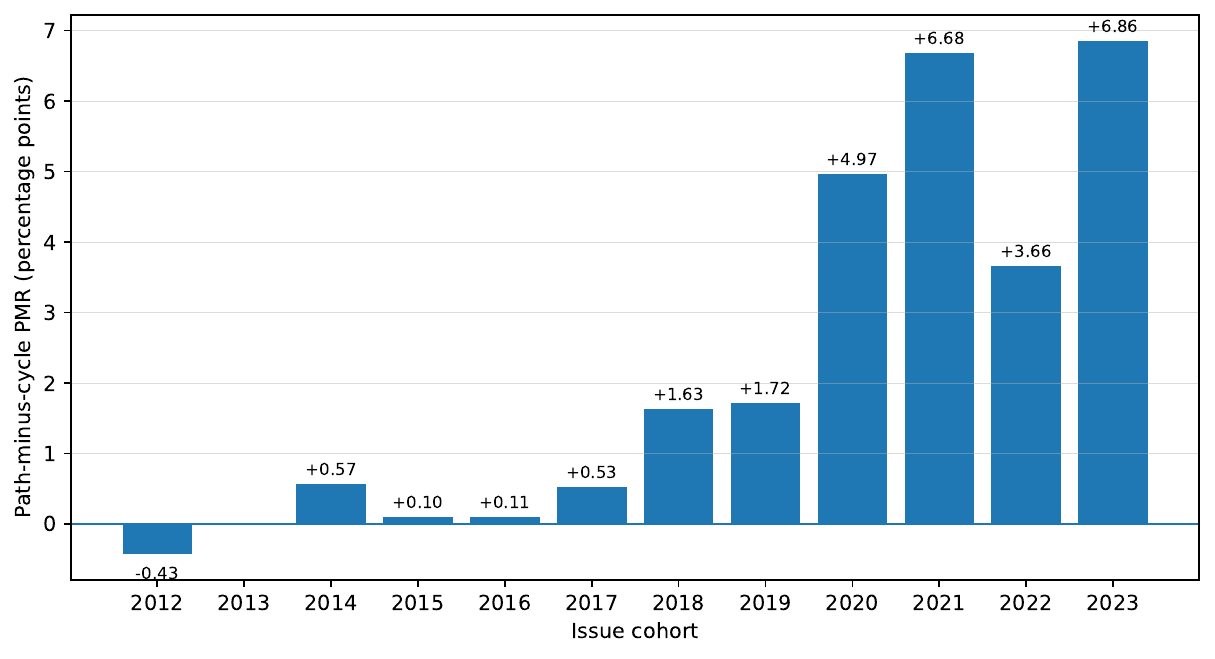}
\caption{Terminal path-minus-cycle PMR difference. The 2012 reversal and the near-zero 2013 difference are retained rather than absorbed into the aggregate.}
\label{fig:annual-advantage}
\end{figure}

\begin{figure}[htbp]
\centering
\includegraphics[width=0.90\textwidth]{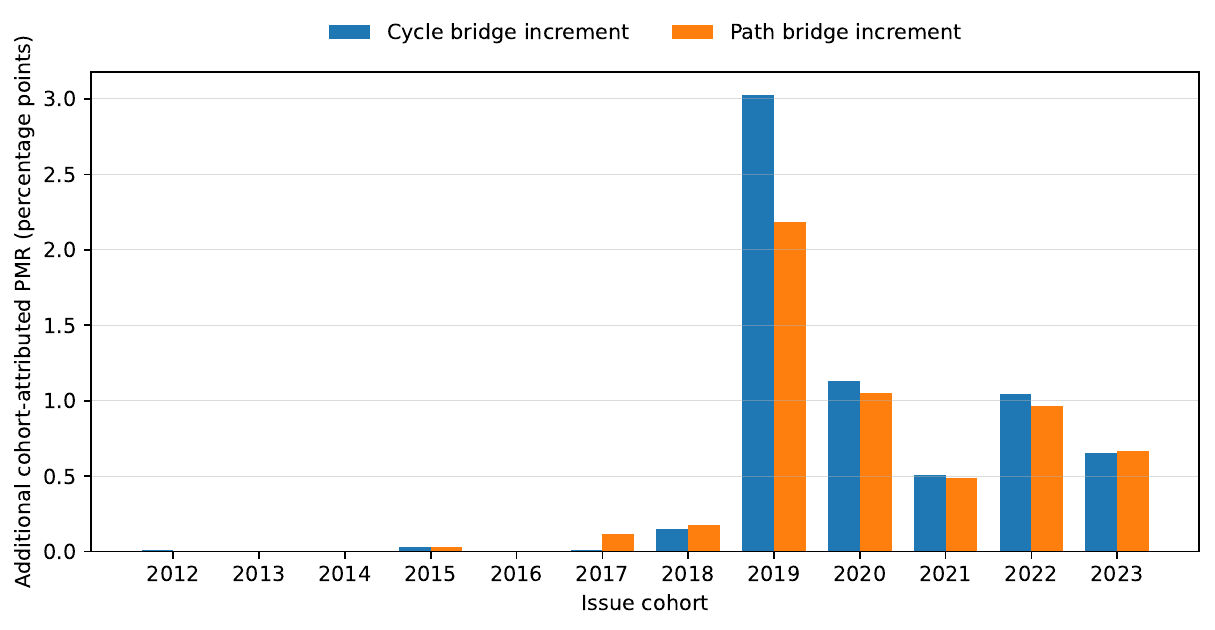}
\caption{Additional cohort-attributed PMR realized during each terminal bridge.}
\label{fig:bridge-increment}
\end{figure}

\subsection{Uninterrupted full-horizon robustness experiment}\label{sec:continuous-horizon}

The annual comparison is useful because it exposes repeated cohort-level heterogeneity, but its prescribed terminal bridges still impose accounting boundaries. To test whether those boundaries materially create the path advantage, the two CDG policies were rerun on one uninterrupted causal stream containing every analytical invoice issued from 1 January 2012 through 31 December 2023. There are no annual resets, cohort closures, or intermediate bridge cutoffs in this experiment: each source UID is introduced once on its issue date, any consumed amount is permanently removed, and every residual remains in the physical state until it is consumed or ceases to be temporally active. The observed 1 January--8 February 2024 records are introduced only after the 2012--2023 issue horizon as a terminal bridge. They may support clearing, but their \EUR0.630 billion of face value is excluded from the \EUR99.705 billion analytical denominator.

Table~\ref{tab:continuous-horizon} reports the result. At 31 December 2023, before the terminal bridge, cycle netting with $L\leq8$ has reduced 43.720\% of 2012--2023 issue mass and path-enabled clearing 48.457\%. The observed 2024 bridge contributes a further \EUR0.145 billion and \EUR0.148 billion, respectively, when only the PMR attributable to 2012--2023 source fragments is credited. Terminal PMR is therefore 43.865\% for cycle netting and 48.605\% for path clearing, a path advantage of \EUR4.727 billion or 4.741 percentage points.

\begin{table}[htbp]
\centering
\caption{Uninterrupted 2012--2023 causal stream with observed 2024 records used only as a terminal bridge. The denominator contains 749,952 invoices and \EUR99.705 billion issued in 2012--2023; 2024 bridge mass is excluded.}
\label{tab:continuous-horizon}
\footnotesize
\renewcommand{\arraystretch}{1.14}
\setlength{\tabcolsep}{4.5pt}
\begin{tabularx}{\textwidth}{>{\raggedright\arraybackslash}Xrrr}
\toprule
Measurement & Cycle $L\leq8$ & Path-enabled & Path advantage \\
\midrule
PMR at 31 Dec. 2023 (\%) & 43.7195 & 48.4571 & 4.7376 pp \\
PMR at 31 Dec. 2023 (EUR bn) & 43.591 & 48.314 & 4.724 \\
2024 bridge PMR attributed to 2012--2023 (EUR bn) & 0.145 & 0.148 & 0.003 \\
Terminal PMR (\%) & \textbf{43.8648} & \textbf{48.6054} & \textbf{4.7406 pp} \\
Terminal PMR (EUR bn) & \textbf{43.735} & \textbf{48.462} & \textbf{4.727} \\
\bottomrule
\end{tabularx}
\end{table}

For non-bilateral paths crossing the 2023--2024 boundary, the primary estimate assigns PMR symmetrically to the two consumed source fragments. Conservative and liberal attribution bounds give terminal path PMR of 48.596\% and 48.614\%, so the boundary-allocation convention is immaterial to the comparison. Relative to the rolling annual design, removing intermediate cohort boundaries raises both policies' aggregate PMR: cycle netting increases from 43.347\% to 43.865\%, and path clearing from 48.202\% to 48.605\%. The path advantage narrows only from 4.855 to 4.741 percentage points. Thus, the principal empirical difference is not an artifact of annual segmentation or of the two-month bridge construction; if anything, allowing residual obligations to persist across the entire horizon benefits the bounded-cycle policy slightly more.

\subsection{Representative cumulative PMR trajectories}\label{sec:curves}

Terminal values conceal when the difference emerges. Figures~\ref{fig:cumulative-early} and \ref{fig:cumulative-late} show unsmoothed cohort-attributed PMR. The vertical dashed line marks the start of the terminal bridge. In 2012, cycle netting overtakes the path policy in a small and sparse network. In 2020 and in the highly reciprocal 2022 network, both policies continue clearing prior-cohort obligations after 31 December while the path lead persists. PMR generated by mixed-year operations is partitioned among issue cohorts, and the underlying invoice fragments are consumed only once.

\begin{figure}[p]
\centering
\begin{subfigure}{0.92\textwidth}
\centering
\includegraphics[width=\linewidth]{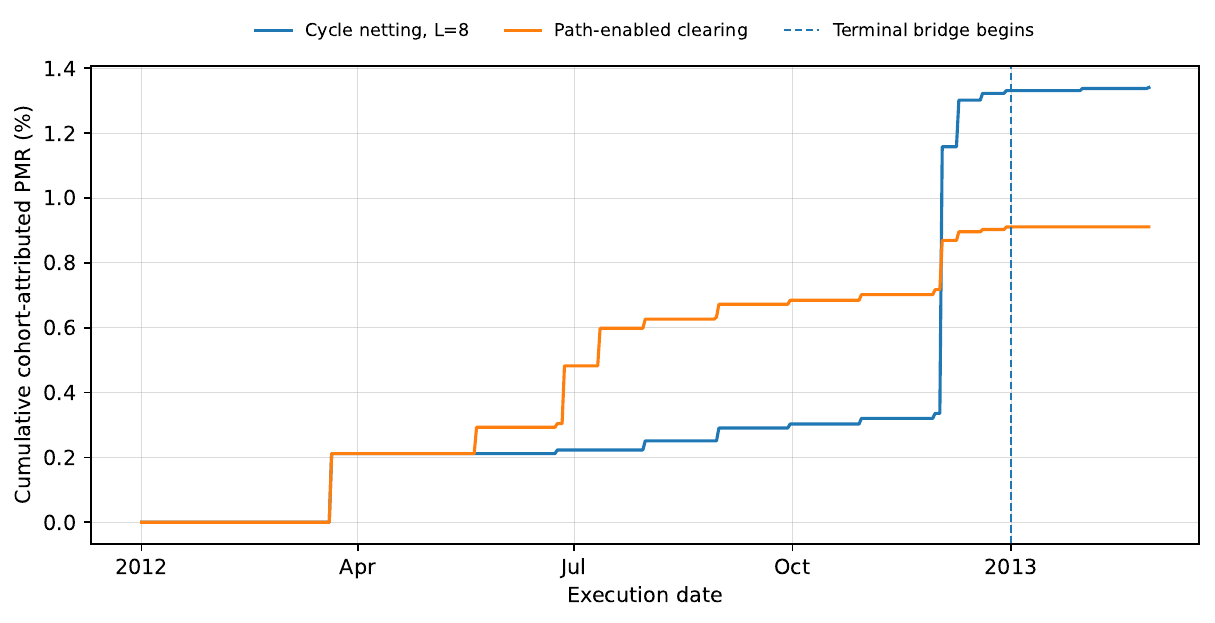}
\caption{2012: small-cohort reversal.}
\end{subfigure}
\par\vspace{2mm}
\begin{subfigure}{0.92\textwidth}
\centering
\includegraphics[width=\linewidth]{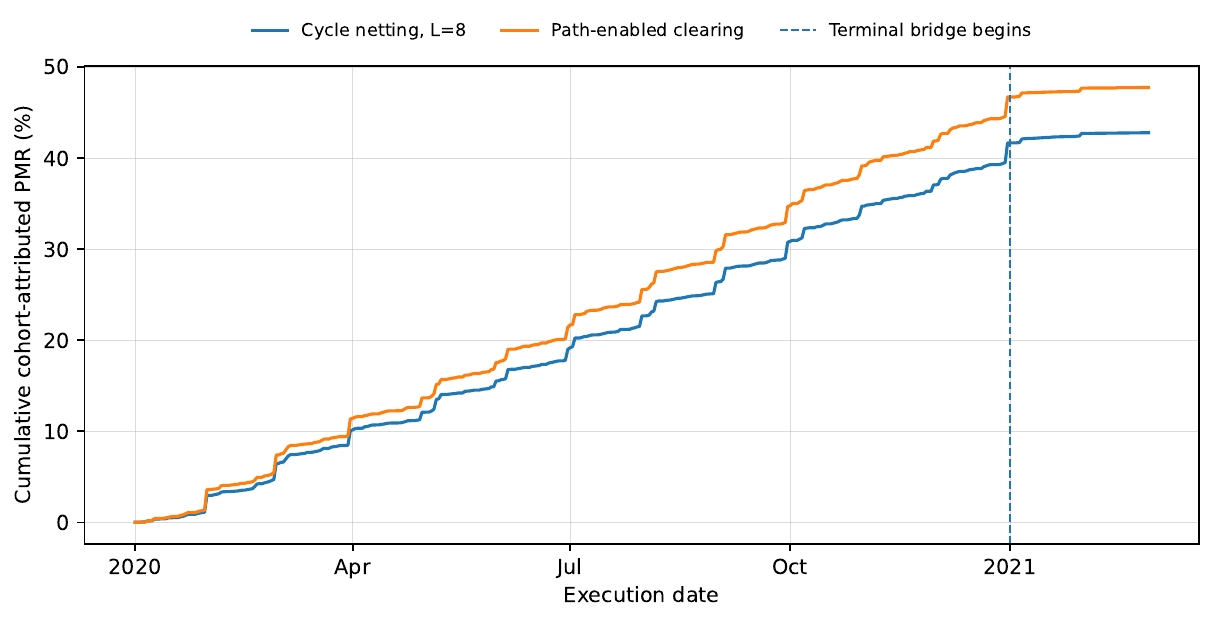}
\caption{2020: sustained path advantage in the largest record-count cohort.}
\end{subfigure}
\caption{Representative cumulative issue-cohort PMR trajectories for 2012 and 2020. Legends are placed above the plotting areas.}
\label{fig:cumulative-early}
\end{figure}

\begin{figure}[p]
\centering
\begin{subfigure}{0.92\textwidth}
\centering
\includegraphics[width=\linewidth]{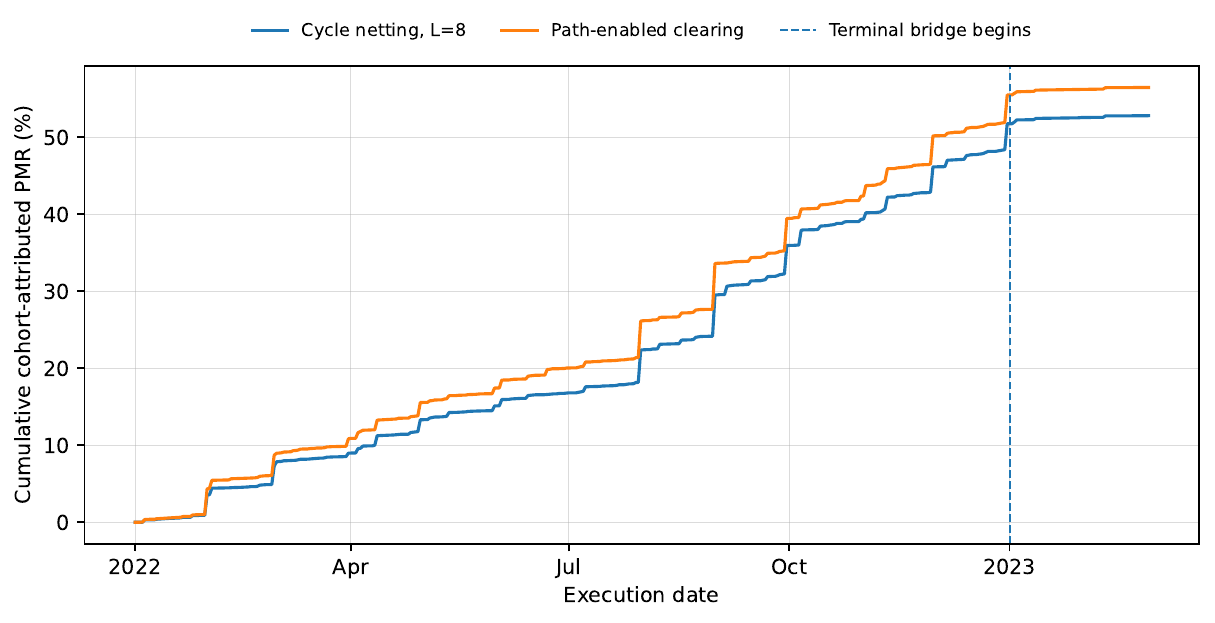}
\caption{2022: path advantage in the cycle-saturated network.}
\end{subfigure}
\par\vspace{2mm}
\begin{subfigure}{0.92\textwidth}
\centering
\includegraphics[width=\linewidth]{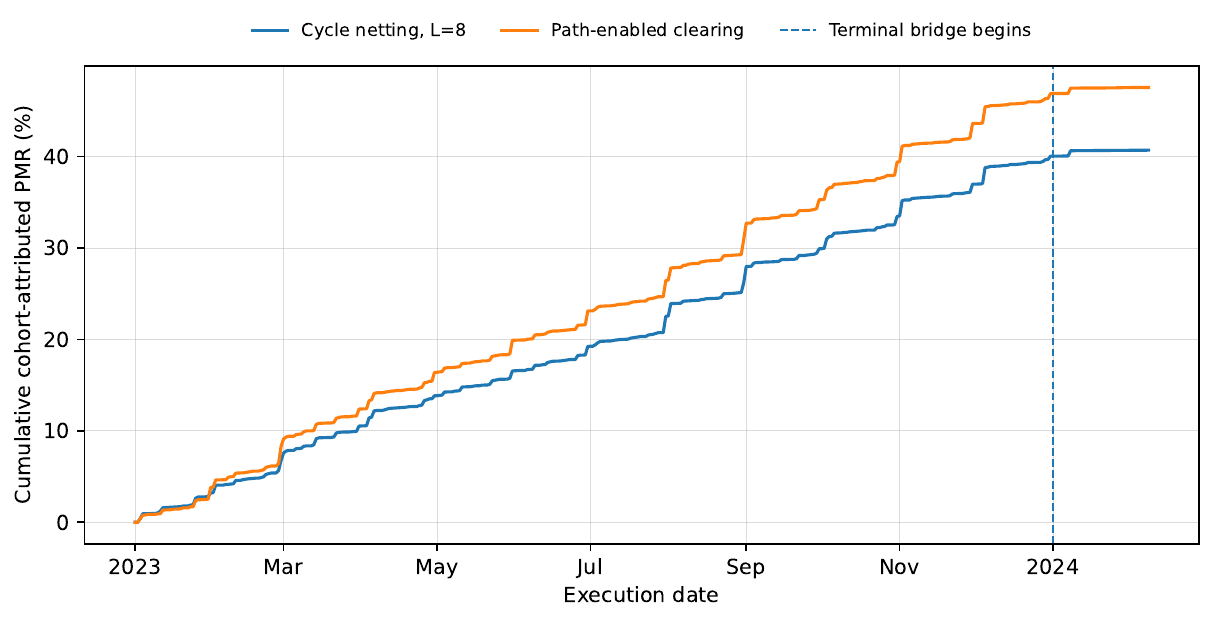}
\caption{2023: path advantage in the final issue cohort.}
\end{subfigure}
\caption{Representative cumulative issue-cohort PMR trajectories for 2022 and 2023.}
\label{fig:cumulative-late}
\end{figure}

\subsection{Full-information bounds and heuristic gaps}\label{sec:lp-results}

Continuous full-information LPs were solved on 156 small weakly connected components and 100 induced strongly connected core samples. In all 256 instances, the path LP bound weakly dominates the corresponding path CDG value, and the cycle LP bound weakly dominates cycle CDG, confirming Proposition~\ref{prop:dominance}. Across the 256 instances, the path-only LP bound exceeds the cycle-only bound in 149, the cycle bound exceeds the path bound in three, and they are equal in 104. On the 100 core samples, a mixed path--cycle LP improves on the better single-move-set bound in 29.

\begin{figure}[htbp]
\centering
\includegraphics[width=0.93\textwidth]{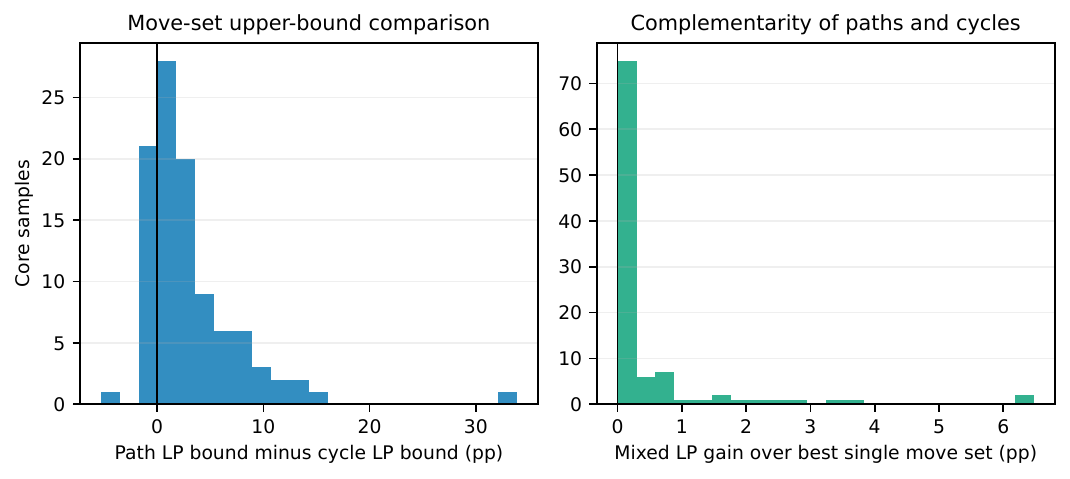}
\caption{Full-information continuous LP comparisons on 100 induced strongly connected core samples. The mixed formulation documents complementarity rather than universal dominance of either move set.}
\label{fig:lp-bounds}
\end{figure}

Heuristic gaps are material. On the 100 core samples, the median path LP gap above CDG is 1.31\% of sample invoice mass and the mean is 5.38\%; the corresponding cycle values are 0 and 0.90\%. CDG is therefore an operational benchmark, not an optimization bound. The LP results also show why the principal empirical claim should remain policy-level: path capacity is usually larger on these samples, but not universally, and a mixed move set can dominate either primitive alone.

\subsection{Comparator robustness: mixed policies, ordering, and cycle length}\label{sec:robustness-comparator}

The supplementary comparator experiments in this subsection stop at 31 December and are interpreted as issue-year scheduling diagnostics rather than bridge-adjusted headline estimates. A sequential cycle-then-path CDG policy modestly improves the path-only result in the evaluated years. In 2019 it reaches 13.894\%, compared with path-only 13.665\%. In 2022 it reaches 55.878\%, compared with path-only 55.514\%. This improvement and the LP evidence support hybrid proposal engines: the move sets share capacity but are not strict substitutes.

Randomized near-greedy schedules show both stability and order sensitivity. Across 500 runs for each year-regime combination, path-minus-cycle PMR is positive in all 2019 samples, 70\% of 2020 samples, 80--90\% of 2022 samples depending on regime, and 60--70\% of 2023 samples. These are induced strongly connected samples rather than annual replications; the dispersion therefore qualifies the deterministic comparison rather than overturning it.

Cycle-length sensitivity is small by $L=8$ in 2022 and 2023. In 2022, CDG cycle PMR rises from 51.614\% at $L=4$ to 51.786\% at $L=8$ and 51.789\% at $L=10$. In 2023 calendar-only execution, the corresponding values are 38.316\%, 38.381\%, and 38.381\%.

\begin{figure}[htbp]
\centering
\includegraphics[width=0.88\textwidth]{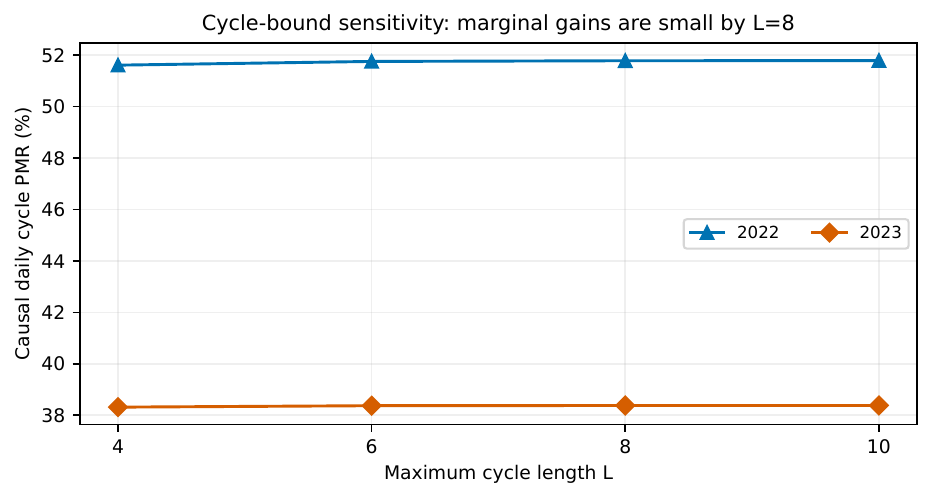}
\caption{Cycle-length sensitivity in 2022 and 2023. Candidate counts grow sharply, whereas incremental CDG PMR beyond length eight is negligible in these cohorts.}
\label{fig:cycle-length}
\end{figure}

\subsection{Heterogeneity, SCC location, and the 2022 case}\label{sec:heterogeneity}

The path result is not explained solely by open chains. Within the issue-year phase, 86.96\% of 2019 path PMR, 95.20\% of 2020, 99.41\% of 2022, and 96.30\% of 2023 occurs inside strongly connected components. Across the corresponding bridge-adjusted outcomes, reciprocal cancellation accounts for 66.66\%, 69.82\%, 80.73\%, and 74.49\%, respectively. Reciprocal paths are economically important because they yield $2q$ PMR without an instruction, but non-bilateral paths still contribute material PMR and are the distinctive redirection mechanism.

The cycle-saturated 2022 graph is therefore informative. Almost all local paths lie in the cyclic core, yet after its complete bridge the path policy still exceeds the cycle policy by 3.661 percentage points under CDG. The advantage arises from operation granularity, reciprocal two-edge opportunities, different use of bottleneck records, and temporal sequencing, not merely access to an acyclic periphery.

Exploratory regressions on 100 core samples do not identify a single strong univariate predictor. Spearman correlations between path advantage and candidate-value ratio, reciprocity share, cycle closure, temporal persistence, edge-weight concentration, candidate-conflict concentration, or records per edge are small and statistically weak. A heteroskedasticity-robust multivariate model with year controls has $R^2=0.196$. These results argue against a one-number topology explanation and motivate richer conflict-graph and temporal-capacity models.

Component resampling further shows concentration. The largest weakly connected component contains 49.7\% of 2019 mass and more than 95\% in 2022--2023. Removing it nearly eliminates or reverses the annual path advantage. Bootstrap intervals across components consequently cross zero. This is not evidence that the annual result is spurious; it shows that effective inference is governed by dominant connected systems rather than by hundreds of thousands of nominally separate invoices.

\begin{figure}[htbp]
\centering
\includegraphics[width=0.84\textwidth]{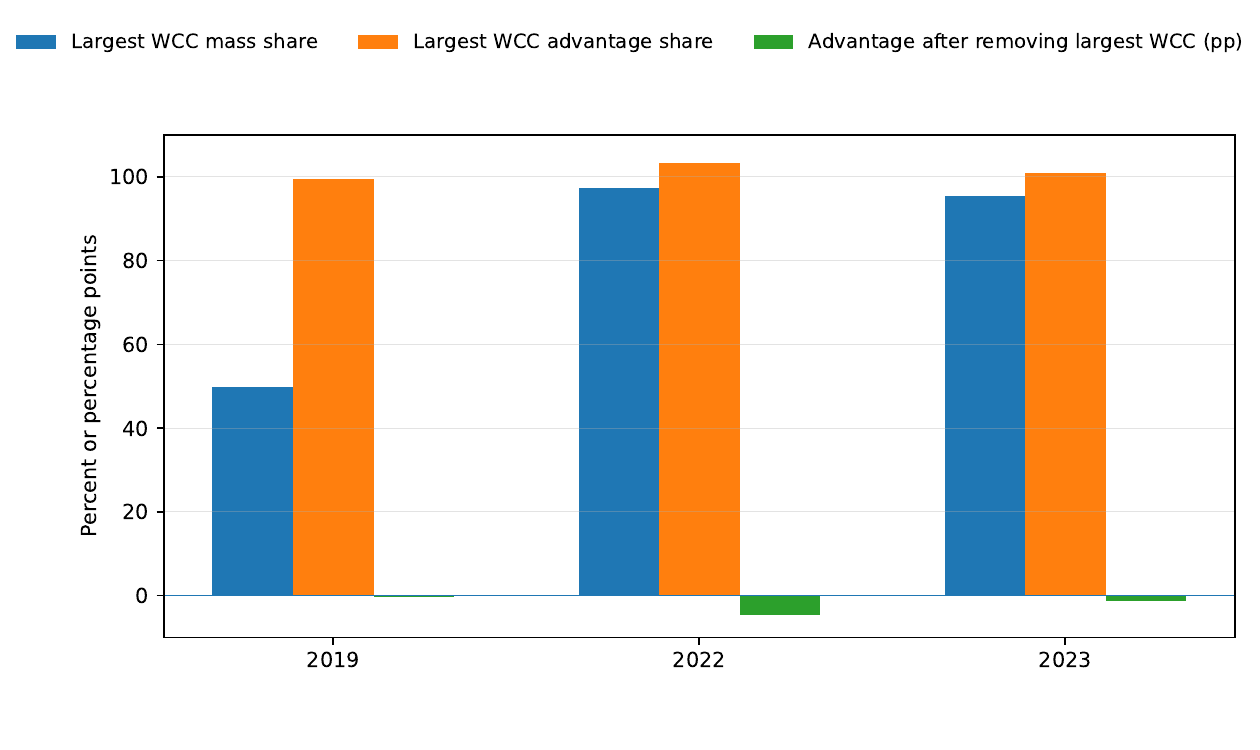}
\caption{Concentration of mass and path advantage in the largest weakly connected component. Component bootstrap intervals are interpreted as stability diagnostics, not IID confidence intervals.}
\label{fig:component-concentration}
\end{figure}

\subsection{Payer acceleration and constrained feasibility}\label{sec:acceleration-results}

Path clearing can accelerate the first payer because the instruction inherits the earlier of the two source due dates. In standalone issue-year CDG diagnostics used to construct the acceleration frontier, amount-weighted mean acceleration on non-bilateral instruction mass is 15.25 days in 2020, 12.24 days in 2022, and 7.90 days in 2023. The positive-only means are 34.60, 17.25, and 24.13 days. Acceleration is concentrated: the top ten payers account for approximately 40.1\%, 82.2\%, and 57.3\% of accelerated instruction mass in those years.

Figure~\ref{fig:acceleration-frontier} gives a post-hoc feasible retention curve: it removes operations from the unconstrained log when their paired fragments exceed a stated acceleration cap, without reoptimizing the remaining schedule. It is therefore a conservative operational diagnostic, not the optimal PMR under each cap. A zero-day cap retains 86.7\% of unconstrained daily path PMR in 2020, 86.2\% in 2022, and 91.7\% in 2023. At 10 days, retention rises to 90.6\%, 93.3\%, and 94.7\%.

\begin{figure}[htbp]
\centering
\includegraphics[width=0.86\textwidth]{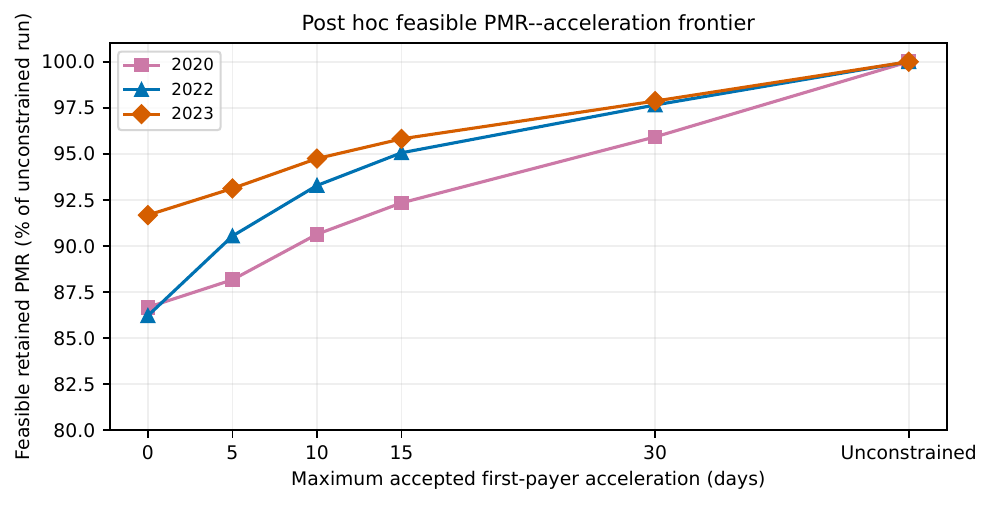}
\caption{Post-hoc feasible PMR retained after applying first-payer acceleration caps to the unconstrained path logs. The curves do not reoptimize the schedule.}
\label{fig:acceleration-frontier}
\end{figure}

These results make acceleration a central governance constraint rather than a secondary note. A deployable objective should maximize expected accepted PMR net of acceleration compensation, exposure, communication, and failure costs. The main unconstrained results quantify algorithmic capacity, not adoption-adjusted welfare.

\subsection{Workload, validation, and provenance sensitivities}\label{sec:workload}

Across the rolling sequence, the cycle policy performs 67,425 operations and the path policy 158,262. The \EUR4.841 billion aggregate difference corresponds descriptively to approximately \EUR53,290 per additional path operation. This is not a transaction-cost threshold because operations differ in participant count, record fragments, acceptance probability, exposure, and legal work.

Validated-only and exact-duplicate-filtered sensitivities preserve the direction in the principal large experiments. Conservative and liberal mixed-year attribution place aggregate path PMR between 47.844\% and 48.560\%, around the symmetric estimate of 48.202\%; both remain above the 43.347\% cycle result. Phase replay returns zero temporal violations, bad fragment sums, overconsumption, negative residuals, PMR identity gaps, and firm-level net-position errors. The cross-phase UID audit additionally verifies that no bridge invoice is introduced or consumed twice.

\section{Interpretation and implications}\label{sec:discussion}

\subsection{What the evidence establishes}\label{sec:establishes}

The evidence supports three conclusions. First, the path move set is operationally meaningful under source-record-exact temporal discipline: its advantage is not created by annual aggregation or by counting both removed invoice legs as economic relief. Second, under the same causal daily information and the same next-year bridge design, the evaluated path policy is empirically stronger than complete-candidate cycle netting with $L\leq8$ in most annual cohorts. The uninterrupted 2012--2023 robustness experiment strengthens that interpretation: after removing intermediate cohort closures entirely and using 2024 records only as terminal bridge support, path PMR remains 48.605\% versus 43.865\% for cycles. PMR-aligned scoring, non-reusable source records, and the continuous-horizon result therefore rule out reciprocal-path score asymmetry, a hard 31 December cutoff, bridge-balance restoration, or annual segmentation as explanations for the aggregate advantage. Third, neither result establishes global dominance. The LP samples contain a small number of cases in which the cycle-only bound is higher, and mixed bounds or mixed heuristics can improve on either move set alone.

CDG is deliberately the sole annual execution regime because it answers the operational question cleanly: what can each move set realize when both see exactly the same invoices as they arrive? The full-information LP is a separate theoretical diagnostic, not a competing implementation. This hierarchy avoids conflating causal deployability, greedy ordering, and optimization.

The 2022 case extends the interpretation beyond open-chain reach. Path clearing remains valuable in a network whose local motifs are almost entirely cyclic. Two-edge operations can still use common-day capacity differently, exploit reciprocal cancellation, avoid long-coalition bottlenecks, and leave a different residual state. Static cycle saturation is therefore not a sufficient predictor of move-set performance.

\subsection{Economic and legal applicability}\label{sec:legal}

PMR is a reduction in gross post-instruction settlement burden. It is not new liquidity, profit, welfare, or solvency improvement. Financing remains necessary for unmatched net obligations. A combined platform could apply consented clearing first and finance residual claims second, but the benefits would depend on fees, funding costs, acceptance, and risk.

The empirical accounting is legally meaningful only under a specified discharge mode. Under consented discharge, redirected payment extinguishes the matched source portions and PMR is realized after performance. Under instruction-only mode, PMR is proposed until payment finality. Under novation, the original claims may be replaced earlier, but counterparty credit exposure changes. Production deployment must define assignment or novation rules, failed-payment reversal, tax treatment, insolvency priority, sanctions screening, KYC responsibilities, dispute rights, and audit retention.

The method is potentially applicable wherever organizations hold overlapping dated payables and receivables under a common legal or contractual framework: supplier networks, public-sector arrears, telecommunications settlements, platform ecosystems, intercompany treasury, and cooperative purchasing networks. Applicability requires reliable identifiers, outstanding balances, legal eligibility, and participant authorization. The current data use invoice face value because a consistently interpretable outstanding-balance field is unavailable; this limits claims about immediately executable cash obligations.

\subsection{Path-enabled clearing as an agentic research agenda}\label{sec:decentralized-agenda}

The empirical policy and the future implementation architecture must remain conceptually separate. CDG is a synchronized centralized benchmark: on each day it observes the active platform graph, ranks candidates, executes a fixed point, and carries residual records forward. The path operation is the economic primitive beneath that scheduler. A decentralized system need not reconstruct the global CDG ranking. It can preserve the same atomic records, common-day predicate, accounting identities, and replay evidence while replacing daily global search with asynchronous local discovery, reservation, consent, and commit.

The path primitive has a structural reason to be attractive in that setting. An intermediary $B$ can identify a proposal $A\rightarrow B\rightarrow C$ from two incident relationships. The evidence required to verify the amount is confined to source fragments on those two edges, and the authorization coalition contains at most three firms. This scope does not grow with the surrounding network. A length-$k$ circuit, by contrast, requires closure discovery, evidence for $k$ relationships, and potentially simultaneous agreement by $k$ firms. Under partial participation and short-lived capacity, every additional participant can create another rejection, timeout, stale reservation, or disclosure boundary. The claim established here is therefore one of bounded locality and coalition size. Whether it produces lower end-to-end communication cost or higher realized PMR is an empirical question for the next study.

The modern agentic stack suggests how that study can be organized. The autonomous-supply-chain methodology of Xu and co-authors provides the domain pattern: firms are represented by specialized agents that act under local objectives and organizational controls rather than by one unconstrained general model \cite{xu2024implementation}. The present paper adds a financial state transition to that pattern. Each agent can expose only eligible invoice commitments and private acceptance rules, while the system's economic outcome remains defined by source-record conservation.

A2A and MCP occupy complementary technical layers. A2A can advertise a firm's clearing capability, carry a path proposal as a long-running task, and communicate status changes without requiring counterparties to reveal their internal models or tools \cite{a2a2025}. MCP can give a local agent controlled access to an invoice vault, credit policy, sanction screen, or deterministic capacity calculator \cite{mcp2025}. Neither protocol should decide monetary capacity. Their role is to transport requests and expose authorized functions; the common-day verifier and residual ledger remain deterministic.

AP2 offers a useful precedent for commitment. Its mandates and receipts show how delegated intent, authorization scope, and dispute evidence can be represented separately from conversational reasoning \cite{ap22026}. An invoice-clearing implementation would require different legal objects---for example, source-fragment commitments, acceleration mandates, discharge conditions, and an instruction receipt---but the same architectural principle applies: the agent proposes, while signed credentials and verifiers authorize and record the economic action. AgenticPay supplies a complementary evaluation lesson. By modelling negotiation under private constraints and measuring feasibility, efficiency, welfare, timeouts, and violations, it demonstrates why an agentic clearing experiment must evaluate more than agreement rate or linguistic quality \cite{agenticpay2026}.

Table~\ref{tab:agentic-stack} maps these modern components to the research contribution enabled by the present paper. The contribution is not another general agent framework. It is a deterministic, auditable economic kernel that an agent ecosystem currently lacks: a canonical invoice state, an exact common-day proposal amount, a net-position-preserving transition, issue-cohort attribution, and replay tests that can reject invalid or duplicated commitments.

\begin{table}[htbp]
\centering
\caption{Modern agentic components and the invoice-clearing research agenda enabled by this paper.}
\label{tab:agentic-stack}
\footnotesize
\renewcommand{\arraystretch}{1.13}
\setlength{\tabcolsep}{3.5pt}
\begin{tabularx}{\textwidth}{>{\raggedright\arraybackslash\bfseries}p{0.18\textwidth}>{\raggedright\arraybackslash}p{0.25\textwidth}>{\raggedright\arraybackslash}X>{\raggedright\arraybackslash}p{0.22\textwidth}}
\toprule
Layer & Modern reference & Contribution from the present model & Open research question \\
\midrule
Domain agents & Autonomous supply-chain MAS \cite{xu2024implementation} & Firm-local invoice state, policy constraints, and auditable financial actions & Participation, incentives, organizational control, and cross-firm governance \\
Agent messaging & A2A \cite{a2a2025} & A typed path-proposal and status lifecycle with bounded three-firm scope & Discovery quality, message volume, latency, and stale tasks \\
Tool and data access & MCP \cite{mcp2025} & Controlled access to invoice vaults, policy tools, and deterministic verifiers & Least privilege, prompt injection, credential isolation, and tool trust \\
Authorization and evidence & AP2 \cite{ap22026} & Clearing-specific mandates, fragment reservations, discharge receipts, and reversal evidence & Legal enforceability, finality, replay resistance, and liability \\
Negotiation and evaluation & AgenticPay \cite{agenticpay2026} & Private acceleration, exposure, fee, and compensation constraints with PMR-based metrics & Strategic behavior, fairness, adoption, and welfare under language-mediated bargaining \\
\bottomrule
\end{tabularx}
\end{table}

A candidate event-driven protocol would proceed through six auditable stages. First, a firm advertises clearing capabilities and admissible policy ranges without exposing its invoice book. Second, an intermediary or federated broker discovers a two-edge opportunity. Third, each edge owner returns signed source-fragment commitments and a short-lived reservation. Fourth, agents negotiate only the discretionary terms---acceleration, compensation, exposure, and discharge mode---while the amount remains bounded by deterministic common-day capacity. Fifth, a verifier checks signatures, record activity, residual sufficiency, sanctions and policy constraints, then performs an atomic commit that either consumes both fragments and creates the instruction or changes nothing. Sixth, all parties receive a receipt suitable for replay, reconciliation, timeout, or reversal.

This architecture gives language models a deliberately narrow role. They may interpret human policies, explain alternatives, search for counterparties, or negotiate compensation. They should not calculate the authoritative amount, mutate residual balances, waive a legal constraint, or certify conservation. The separation addresses a central weakness identified by the modern multi-agent survey literature: coordination quality cannot substitute for grounded state and verifiable actions \cite{guo2024multiagents}.

The decisive future experiment should run four mechanisms on the same invoice-event stream: centralized path CDG, asynchronous path agents, asynchronous cycle agents, and a hybrid market able to propose either move. Treatments should vary participation, message delay, reservation lifetime, rejection, private acceleration caps, counterparty exposure, strategic withholding, identity and privacy overhead, and failure recovery. Outcomes should include accepted PMR, PMR-days, messages per accepted euro, disclosed relationships, stale reservations, time to finality, acceleration distribution, concentration, fairness, and distance from the tractable full-information bound.

The paper supplies both references needed to interpret that experiment. CDG is the causal centralized baseline: it measures what the move set can achieve before communication and consent frictions. The full-information LP is the oracle reference on tractable instances: it measures remaining allocation loss. The gap between agent execution and CDG would estimate coordination and participation losses; the gap between CDG and the LP would measure centralized heuristic and foresight losses. This decomposition is a direct contribution of the present empirical design to the broader agentic-AI agenda.

The principal hypothesis is that bounded three-firm paths will retain a larger share of centralized PMR than longer cycles as participation falls or consent becomes less reliable. Secondary hypotheses concern fewer stale reservations, smaller disclosure scope, and faster finality. These advantages are not guaranteed: repeated local search, strategic withholding, adverse selection, or mandate overhead may outweigh coalition-size benefits. A rigorous test of that trade-off---rather than a prototype demonstration alone---is the appropriate next research contribution.

\section{Limitations and future research}\label{sec:limitations}

The study has several limitations. The annual cohorts are heterogeneous in source coverage, scale, provenance, and terminal follow-up. The mass-weighted result is not a population estimate or national time trend. The bridge design reduces rather than eliminates right-censoring: unresolved prior-cohort records are closed after their observed follow-up even when contractual due dates extend further. The sequence also begins in 2012 without a pre-2012 opening state. The 2021 issue cohort uses a separate accepted workbook pipeline, with documented crosswalks at its boundaries. Raw identifiable invoices cannot be openly released, which limits independent empirical replication even though software, derived tables, and replay procedures can be shared.

The large policies remain greedy. The LP analysis is limited to small and induced subgraphs and uses continuous divisibility. The uninterrupted full-horizon experiment removes intermediate annual closures and shows that the main path advantage persists, but the sequence still begins in 2012 without a pre-2012 opening state and ends with the observed 2024 follow-up. Future work should develop scalable temporal path--cycle column generation, exact mixed-integer models on larger instances, approximation bounds, and scheduling policies that explicitly reserve capacity for future arrivals. Longer continuous panels and alternative terminal follow-up horizons would further quantify residual start- and end-boundary sensitivity.

The path instruction layer is not recycled. Recursive reuse could increase compression but would introduce additional exposure, legal chaining, and termination questions. It should be analyzed as a distinct model. Likewise, the current mixed heuristic is sequential; integrated candidate ranking and joint optimization may perform better.

Acceleration sensitivities are post-hoc rather than reoptimized frontiers. Future policies should enforce payer-level caps during candidate selection, include compensation and utility, and report adoption-adjusted PMR. Risk limits, counterparty whitelists, sanctions, currency, partial payment, dispute, and legal eligibility are not modeled.

Statistical inference is constrained by dependence and network concentration. Annual sign tests have only 12 cohorts; component bootstraps are dominated by giant connected components; induced samples do not recreate the full graph. Future data partnerships should provide additional jurisdictions, sectors, years, and stable identifiers, enabling hierarchical analysis across genuinely independent networks.

Finally, decentralized execution is not evaluated here. The fixed three-firm scope of a path proposal motivates the event-driven research design in Section~\ref{sec:decentralized-agenda}, but it does not prove lower end-to-end communication cost, higher acceptance, stronger privacy, resilience, incentive compatibility, or welfare. Those properties must be measured under partial participation, asynchronous negotiation, strategic behavior, privacy proofs, failure recovery, and legally valid authorization.

\section{Conclusion}\label{sec:conclusion}

This paper develops an atomic common-day invoice-clearing model and evaluates path and bounded-cycle policies across temporal invoice networks under one causal daily schedule. Fixed-candidate capacity is exact: $\delta_F^*(s)=\max_t\min_{e\in F}c_e(t;s)$. Every operation is assigned to a common day and consumed from active source records. Path accounting preserves net positions on the combined invoice-plus-instruction state and avoids conflating invoice compression with economic PMR.

The empirical design follows invoice state continuously across annual boundaries. Each source UID enters once, cumulative consumption cannot exceed face value, and only residual bridge balance proceeds into the invoice's own issue year. Mixed-year PMR is partitioned from the consumed source fragments, so physical state and cohort accounting remain additive. Independent replay verifies temporal activity, fragment sums, instruction mass, PMR identities, nonnegative residuals, and firm-level net positions.

Across 749,952 invoices and \EUR99.705 billion issued from 2012 through 2023, path-enabled clearing reaches 48.202\% PMR and complete-candidate cycle netting with $L\leq8$ reaches 43.347\%. The difference is \EUR4.841 billion and 4.855 percentage points. Path clearing leads materially in ten cohorts, is practically tied in 2013, and trails in 2012. In cycle-saturated 2022, where almost every local two-edge path lies inside the cyclic core, path PMR remains higher. When all 2012--2023 invoices are instead processed as one uninterrupted causal stream and the observed 2024 records are used only as a terminal bridge, path PMR is 48.605\% versus 43.865\% for cycles, a \EUR4.727 billion or 4.741-point advantage. The result is therefore not explained solely by access to an acyclic periphery or by the annual cohort boundaries used in the principal repeated comparison.

The evidence remains policy-specific. Ordering, reciprocal cancellation, acceleration, network concentration, legal discharge, workload, and source coverage affect outcomes. Tractable full-information programs show both path advantage and path--cycle complementarity. The operational implication is that clearing engines should admit local path transformations alongside cycles and select among them under temporal, consent, exposure, and acceleration constraints; the study does not establish universal move-set or welfare dominance.

The path primitive also opens a focused agentic-AI research agenda. Its fixed two-edge, three-firm scope can be combined with modern agent messaging, controlled tool access, cryptographic mandates, and private-constraint negotiation while retaining deterministic common-day verification and replay. CDG supplies the causal centralized baseline for that future experiment, and the full-information programs supply tractable oracle references. Measuring accepted PMR, communication, disclosure, latency, failure, fairness, and incentive effects would establish whether local path clearing provides a practical decentralized advantage rather than merely an architectural one.

\section*{Data and reproducibility materials}
A pseudonymized version of the invoice records and the deposited reproducibility outputs supporting the experiments are publicly available in Mendeley Data: Peplluis Esteva de la Rosa (2026), \emph{Atomic Common-Day Invoice Clearing: Pseudonymized Invoice Records and Reproducibility Data, 2012--2023}, V1, \doi{10.17632/28rbmvwsm9.1}. Direct company names, VAT/CIF identifiers, invoice numbers, source-file identifiers, and reversible mappings are not included in the public deposit. Identifiable commercial source records are not redistributed.

\appendix
\clearpage

\section{Complete annual source and topology table}\label{app:annual}
\begin{table}[htbp]
\centering
\caption{Annual analytical sequence and selected topology diagnostics.}
\label{tab:annual-topology-full}
\scriptsize
\resizebox{\textwidth}{!}{%
\begin{tabular}{rrrrrrrrr}
\toprule
Year & Records & Mass bn & Firms & Edges & Reciprocity & Largest SCC & $R_{top}$ & $L\leq8$ circuits \\
\midrule
2012 & 339 & 0.069 & 138 & 129 & 0.0620 & 0.0290 & 1.0476 & 7 \\
2013 & 540 & 0.119 & 184 & 158 & 0.0506 & 0.0109 & 1.0000 & 4 \\
2014 & 666 & 0.197 & 196 & 153 & 0.1046 & 0.0153 & 1.4348 & 9 \\
2015 & 1,069 & 0.538 & 256 & 196 & 0.1224 & 0.0078 & 1.1538 & 12 \\
2016 & 2,078 & 0.679 & 400 & 323 & 0.1300 & 0.0075 & 1.6727 & 21 \\
2017 & 4,004 & 0.257 & 601 & 491 & 0.1222 & 0.0067 & 2.2300 & 31 \\
2018 & 7,798 & 0.460 & 1,160 & 1,070 & 0.2093 & 0.0103 & 1.3328 & 125 \\
2019 & 70,155 & 4.424 & 25,383 & 26,358 & 0.0555 & 0.0122 & 1.0080 & 12,479 \\
2020 & 375,386 & 14.020 & 63,371 & 67,600 & 0.0755 & 0.0293 & 1.0001 & 899,363 \\
2021 & 121,844 & 15.153 & 14,250 & 17,116 & 0.2617 & 0.1019 & 2.4252 & 186,945 \\
2022 & 79,471 & 41.576 & 3,743 & 5,883 & 0.6888 & 0.3511 & 1.0002 & 48,109 \\
2023 & 86,602 & 22.212 & 4,460 & 6,409 & 0.5664 & 0.2785 & 1.0000 & 21,674 \\
\bottomrule
\end{tabular}}
\end{table}
\FloatBarrier
\clearpage

\section{Full annual provenance audit}\label{app:provenance}

\begin{table}[htbp]
\centering
\caption{Annual source selection and retained analytical population. Exact fingerprint duplicates are audited separately and retained in the main corpus unless they belong to a known source overlap.}
\label{tab:provenance-full}
\footnotesize
\renewcommand{\arraystretch}{1.12}
\setlength{\tabcolsep}{5pt}
\begin{tabularx}{\textwidth}{rrrr>{\raggedright\arraybackslash}X}
\toprule
Year & Input & Retained & Mass (bn) & Principal coverage or provenance limitation \\
\midrule
2012 & 366 & 339 & 0.069 & Sparse legacy coverage \\
2013 & 587 & 540 & 0.119 & Sparse legacy coverage \\
2014 & 848 & 666 & 0.197 & Sparse legacy coverage \\
2015 & 1,234 & 1,069 & 0.538 & Sparse legacy coverage \\
2016 & 2,376 & 2,078 & 0.679 & Sparse legacy coverage \\
2017 & 4,420 & 4,004 & 0.257 & Sparse legacy coverage \\
2018 & 8,289 & 7,798 & 0.460 & Sparse legacy coverage \\
2019 & 73,108 & 70,155 & 4.424 & Beginning of major coverage expansion \\
2020 & 387,674 & 375,386 & 14.020 & Expanded reporting population \\
2021 & 129,158 & 121,844 & 15.153 & Separate accepted monthly-workbook pipeline \\
2022 & 85,148 & 79,471 & 41.576 & Highly concentrated and reciprocal network \\
2023 & 94,355 & 86,602 & 22.212 & Shorter terminal follow-up in the supplied source \\
\bottomrule
\end{tabularx}
\end{table}

\begin{table}[htbp]
\centering
\caption{Annual cleaning exclusions by reason. Counts are mutually exclusive under the implemented cleaning order.}
\label{tab:cleaning-full}
\scriptsize
\renewcommand{\arraystretch}{1.11}
\setlength{\tabcolsep}{4.3pt}
\begin{tabular}{rrrrrrr}
\toprule
Year & Bad amount & Bad due date & Negative maturity & Cancelled & Missing/self party & Known overlap \\
\midrule
2012 & 12 & 15 & 0 & 0 & 0 & 0 \\
2013 & 27 & 19 & 1 & 0 & 0 & 0 \\
2014 & 32 & 71 & 1 & 78 & 0 & 0 \\
2015 & 45 & 67 & 2 & 51 & 0 & 0 \\
2016 & 72 & 127 & 12 & 87 & 0 & 0 \\
2017 & 145 & 151 & 107 & 13 & 0 & 0 \\
2018 & 317 & 118 & 40 & 16 & 0 & 0 \\
2019 & 2,245 & 346 & 219 & 143 & 0 & 0 \\
2020 & 5,552 & 3,651 & 2,348 & 732 & 5 & 0 \\
2021 & 1,867 & 4,735 & 707 & 0 & 3 & 2 \\
2022 & 1,119 & 3,781 & 502 & 150 & 0 & 125 \\
2023 & 857 & 6,065 & 831 & 0 & 0 & 0 \\
\bottomrule
\end{tabular}
\end{table}
\FloatBarrier
\clearpage

\section{Sequential annual execution and non-reuse controls}\label{app:bridges}

For method $m$, the annual sequence is one physical state stream. January--February records of year $y+1$ are introduced during cohort $y$'s terminal bridge. Their consumed amount is removed immediately and their residual amount becomes the opening state for the remainder of year $y+1$. Older cohort residuals are closed after their prescribed bridge. The implementation stores a global set of introduced UIDs and a cumulative-consumption ledger; duplicate introduction or cumulative consumption above original face value terminates the run.

\begin{table}[htbp]
\centering
\caption{Boundary non-reuse audit. Amounts are EUR billions and are rounded in the table; identities are tested in integer cents. ``Consumed'' refers to physical consumption from newly introduced January--February records during the preceding cohort's bridge; ``carried'' is the residual mass available in the record's own issue year.}
\label{tab:nonreuse-audit}
\footnotesize
\renewcommand{\arraystretch}{1.12}
\setlength{\tabcolsep}{3.2pt}
\begin{tabular}{rrrrrrr}
\toprule
Cohort & New bridge records & New mass & Cycle consumed & Cycle carried & Path consumed & Path carried \\
\midrule
2012 & 82 & 0.002 & 0.000 & 0.002 & 0.000 & 0.002 \\
2013 & 84 & 0.026 & 0.000 & 0.026 & 0.000 & 0.026 \\
2014 & 140 & 0.053 & 0.000 & 0.053 & 0.000 & 0.053 \\
2015 & 189 & 0.083 & 0.000 & 0.082 & 0.000 & 0.082 \\
2016 & 558 & 0.035 & 0.006 & 0.029 & 0.008 & 0.027 \\
2017 & 1,000 & 0.057 & 0.001 & 0.056 & 0.006 & 0.051 \\
2018 & 2,184 & 0.131 & 0.006 & 0.125 & 0.009 & 0.122 \\
2019 & 61,091 & 2.364 & 0.908 & 1.456 & 1.365 & 0.999 \\
2020 & 22,779 & 2.132 & 0.854 & 1.278 & 1.200 & 0.931 \\
2021 & 12,958 & 6.061 & 3.029 & 3.032 & 4.515 & 1.546 \\
2022 & 14,265 & 5.011 & 1.503 & 3.509 & 2.439 & 2.573 \\
2023 & 1,240 & 0.630 & 0.323 & 0.308 & 0.343 & 0.287 \\
\midrule
Total & 116,570 & 16.584 & 6.628 & 9.956 & 9.885 & 6.699 \\
\bottomrule
\end{tabular}
\end{table}

At every boundary and for both methods, new bridge mass equals bridge consumption plus residual carried mass to the cent. All 751,192 source UIDs in the 2012--2024 execution stream are introduced once, and the global consumption ledger reports zero overconsumed records. The final boundary uses its documented observed source horizon; the 2012--2022 subset supplies the equal-horizon sensitivity.

\begin{table}[htbp]
\centering
\caption{Boundary source and coverage by cohort group.}
\label{tab:bridge-coverage}
\footnotesize
\begin{tabularx}{\textwidth}{>{\raggedright\arraybackslash}p{0.14\textwidth}>{\raggedright\arraybackslash}p{0.25\textwidth}>{\raggedright\arraybackslash}p{0.20\textwidth}>{\raggedright\arraybackslash}X}
\toprule
Issue cohort & Next-year source & Coverage & Identifier continuity \\
\midrule
2012--2019 & Harmonized longitudinal source & Complete Jan--Feb $y+1$ & Common anonymous company-code namespace \\
2020 & Accepted 2021 monthly workbooks & Complete Jan--Feb 2021 & One-to-one company crosswalk; unresolved firms remain distinct \\
2021 & Harmonized 2022 source & Complete Jan--Feb 2022 & 2021 CIF-to-anonymous-code crosswalk; unresolved firms remain distinct \\
2022 & Harmonized 2023 source & Complete Jan--Feb 2023 & Common anonymous company-code namespace \\
2023 & Supplied 2024 source & 1 Jan--8 Feb 2024 & Common harmonized company codes; February is incomplete \\
\bottomrule
\end{tabularx}
\end{table}

For cycles, cohort PMR is the sum of removed invoice-leg mass by issue year. Reciprocal path PMR is attributed by removed leg. For a non-bilateral path, the one unit of PMR is split symmetrically across the incoming and outgoing consumed fragments. The conservative bound credits a cohort only when both legs belong to it; the liberal bound credits it when either leg does. Across all cohorts, these bounds place path PMR between 47.844\% and 48.560\%, around the symmetric estimate of 48.202\%.

\section{Optimization and stability diagnostics}\label{app:optimization}

The continuous LP creates one operation variable for each candidate and representative event day, plus source-record allocation variables on every supporting edge. Candidate-edge allocations equal the operation amount, and total allocations from a record cannot exceed its original value. The objective uses PMR per unit: one for non-bilateral paths, two for reciprocal paths, and circuit length for cycles. Because source amounts are divisible in the relaxation, the LP weakly dominates integer-cent schedules.

The 156 small-component instances contain all tractable nontrivial components selected by the screening rule. The 100 induced-core instances comprise 25 samples from each of 2019, 2020, 2022, and 2023, with 3--8 nodes and at most 60 records. The LP is a computational bound for these instances only; it is not extrapolated to the full annual graphs.

\section{Selected replay checks}\label{app:replay}

For each logged operation, replay verifies: (i) every source row exists; (ii) each fragment belongs to the claimed edge; (iii) the operation day lies between issue and due dates; (iv) fragments on each edge sum to the logged amount; (v) residuals remain nonnegative; (vi) instruction mass and PMR identities hold; and (vii) combined firm-level net positions equal the initial positions. Across the 48 method-phase combinations, replay covers 852,387 fragment allocations and reports zero errors. The global UID ledger is separate from phase replay: it accumulates consumption over the full 2012--2024 stream and verifies that every record is introduced once and cumulative consumption never exceeds original amount.

\section{Reference implementation}\label{app:code}

The accompanying Python package implements atomic CSV input, exact edge-day and common-day capacities, PMR-aligned causal daily path and cycle policies, rolling annual execution with non-reusable bridge records, source-fragment issue-year attribution, sequential mixed daily execution, continuous LP bounds for tractable instances, plotting, and independent replay. Unit tests verify temporal admissibility, accounting identities, deterministic fragment consumption, and weak dominance of the causal schedule by the corresponding full-information LP relaxation on synthetic instances. The package contains no raw identifiable invoices or reversible company map.

\begingroup
\sloppy

\endgroup

\end{document}